\documentclass[journal]{IEEEtran}
\usepackage{amsthm}
\usepackage{amsmath}
\usepackage{amssymb}
\usepackage{algorithm}
\usepackage{algorithmic}
\usepackage{graphicx}
\graphicspath{{figures/}}

\newtheorem{mydef}{Definition}
\newtheorem{mytheorem}{Theorem}
\newtheorem{mylemma}{Lemma}

\ifCLASSINFOpdf
\else
\fi
\begin{document}
%
\title{Differentially Private Double Spectrum Auction with Approximate Social Welfare Maximization}
%
%
%

\author{\IEEEauthorblockN{Zhili Chen, Tianjiao Ni, Hong Zhong, Shun Zhang, Jie Cui}\\
\IEEEauthorblockA{School of Computer Science and Technology, Anhui University, Hefei, China\\
Email:{zlchen@ahu.edu.cn, tjni94@163.com, zhongh@mail.ustc.edu.cn,\\
 shzhang27@163.com, cuijie@mail.ustc.edu.cn}}}

\maketitle

\begin{abstract}
Spectrum auction is an effective approach to improving spectrum utilization, by leasing idle spectrum from primary users to secondary users. Recently, a few differentially private spectrum auction mechanisms have been proposed, but, as far as we know, none of them addressed the differential privacy in the setting of double spectrum auctions. In this paper, we combine the concept of differential privacy with double spectrum auction design, and present a Differentially private Double spectrum auction mechanism with approximate Social welfare Maximization (DDSM). Specifically, we design the mechanism by employing the exponential mechanism to select clearing prices for the double spectrum auction with probabilities exponentially proportional to the related social welfare values, and then improve the mechanism in several aspects like the designs of the auction algorithm, the utility function and the buyer grouping algorithm. Through theoretical analysis, we prove that DDSM achieves differential privacy, approximate truthfulness, approximate social welfare maximization. Extensive experimental evaluations show that DDSM achieves a good performance in term of social welfare.
\end{abstract}


%
\IEEEpeerreviewmaketitle

\section{Introduction}
\IEEEPARstart{W}{ith} the rapid development of wireless technologies, radio spectrum has become a scarce resource. Being controlled by governmental agencies (e.g., Federal Communication Commission (FCC) in the US), spectrum is usually distributed with traditional static allocation policy. Nevertheless, researches have shown that the efficiency of this type of spectrum allocation is low in both spatial and temporal dimensions. The major problem of static spectrum allocations is that large amounts of distributed spectrum channels owned by licensed primary users are often idle in geographic areas, while unlicensed secondary users starve for spectrum channels to complete their work. Thus, to alleviate this conflict, spectrum actions \cite{Akyildiz2006NeXt} have emerged to redistribute spectrum channels from licensed primary users to unlicensed secondary users. Due to their perceived fairness and allocation efficiency, spectrum auctions are regarded as an effective market-based approach to mitigating the problem of spectrum scarcity.

In recent years, numerous truthful spectrum auction mechanisms have been proposed. Truthfulness incentivizes bidders to bid their true valuations, and thus avoids market operations. Nevertheless, truthfully bidding also reveals bidders' true valuations to the auctioneer, leading to serious privacy issues. To protect privacy, some privacy-preserving spectrum auction schemes based on techniques of cryptography or secure multi-party computation have been proposed \cite{Huang2013SPRING,Huang2015PPS,Chen2016Towards, Chen2014PS, Chen2017Secure}. These schemes protect privacy from a dishonest auctioneer in the process of auction computations, but normally do not consider the privacy issue by inferring sensitive information from the auction outcome. This later issue can be addressed by spectrum auction mechanisms with differential privacy guarantee, and this is what we focus in this paper.

Typically, the change in only one bidder's bid in an auction might significantly impact on the auction outcome. Conversely, by comparing the outcomes of two auctions whose bid profiles only differ in a single bid, one could infer some sensitive information on bidders' bids. For instance, three bidders $A$, $B$ and $C$ bid $9$, $8$ and $5$, respectively, for a good in a second-price auction. The auction outcome is that the winner is $A$ and its price is $8$. Some time later, the same auction happens again except that $B$ lowers its bid to $0$. Then the auction outcome becomes that the winner is $A$ and its price is $5$. By comparing the two outcomes, $B$ knows that the third price is $5$, which should be kept secret in a second price auction. To address this issue, the concept of differential privacy \cite{Dwork2008Differential,Dwork2006Differential} has been applied to ensure that any such two auctions produce almost identically distributed outcomes, and thus little sensitive information would be leaked by outcome comparisons. In fact, there have been several spectrum auction mechanisms proposed with differential privacy guarantee recently \cite{Zhu2014Differentially,Zhu2015Differentially}. However, they only addressed single-sided spectrum auctions. As far as we know, there is no previous work addressing differential privacy in the setting of double spectrum auctions. To fill this gap, our aim is to design a differentially private mechanism for double spectrum auctions.

In this paper, we propose a \underline{D}ifferentially private \underline{D}ouble spectrum auction mechanism with approximate \underline{S}ocial welfare \underline{M}aximization (DDSM). As shown in Fig.~\ref{fig:framework}, in the auction framework, sellers submit their quotations, and buyers submit their bids to an auctioneer, who performs the auction computations, and announces the auction outcome to both sellers and buyers after the auction. We assume that the auctioneer is trusted, or otherwise we can simulate a trusted auctioneer with secure multi-party computations.  To achieve differential privacy, we employ one exponential mechanism \cite{Mcsherry2007Mechanism} to select the selling clearing price, and the other to select the buying clearing price, exponentially proportional to their related social welfare values, respectively. And then we use both selling and buying clearing prices to determine the winners and allocate spectrum channels. We carefully design our double spectrum auction mechanism such that it also achieves approximate truthfulness and approximate social welfare maximization. Later, we improve our mechanism in several aspects, and further promote the performance of our auction mechanism in term of the social welfare.

\begin{figure}[!t]
\begin{center}
\includegraphics[width=0.48\textwidth]{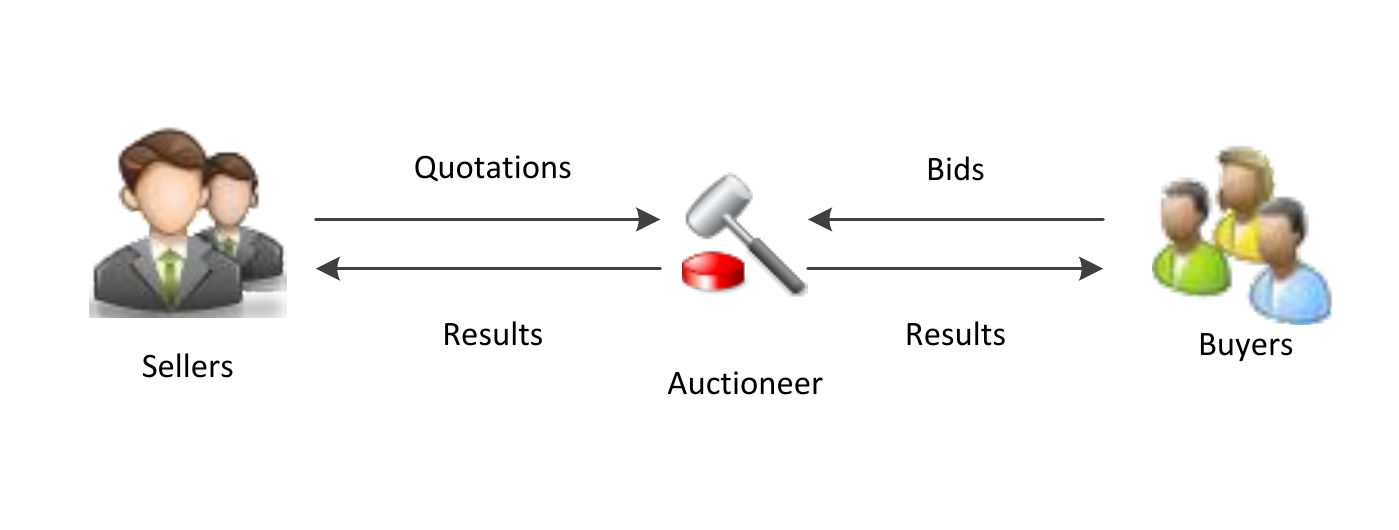}
\caption{Auction framework of DDSM: Sellers submit their quotations and buyers submit their bids to the auctioneer, and the auctioneer returns the results to them.}\label{fig:framework}
\end{center}
\end{figure}

 Our main contributions can be summarized as follows.
\begin{itemize}
   \item We integrate the exponential mechanism with double spectrum auction to achieve differential privacy and approximate social welfare maximization. We also carefully design the auction mechanism to achieve approximate truthfulness. As far as we know, our proposal is the first differentially private mechanism for double spectrum auctions.
    \item We improve our mechanism in the designs of auction algorithm, utility function, and grouping algorithm, such that it achieves a better performance in term of social welfare, and it is applicable to any random buyer grouping algorithm independent on buyers' bids.
    \item We fully implement DDSM, and do extensive experiments to evaluate the performance in term of social welfare. The experimental results show that our mechanism achieves a good performance compared to the double spectrum auction without differential privacy guarantee.
\end{itemize}

The remainder of this paper is organized as follows. Section~\ref{sec:relatedwork} briefly reviews related work and Section~\ref{sec:preliminaries} provides the problem model and introduces some related concepts. We present the detailed design and theoretical proof of DDSM in Section~\ref{sec:ddsm-details}, and improve DDSM in Section~\ref{sec:improvement}. In Section~\ref{sec:experiment}, we implement DDSM, and evaluate the performance in terms of social welfare. In the end, the paper is summarized in Section~\ref{sec:conclusions}.

\section{Related Work}\label{sec:relatedwork}

\textbf{Truthful spectrum auctions.} In recent years, a large number of truthful spectrum auction mechanisms have been proposed. Zhou et al. proposed TRUST \cite{Zhou2009TRUST}, the first truthful double spectrum auction framework enabling spectrum reuse. Jia et al. \cite{Jia2009Revenue} designed an exponential-time truthful spectrum auction mechanism with optimal revenue, and further presented a truthful suboptimal auction mechanism in polynomial time. In \cite{Al2011Truthful}, the author proposed a polynomial-time truthful spectrum auction mechanism with a performance guarantee on revenue. Papers \cite{Xu2011Efficient}, \cite{Zhu2012Truthful} and \cite{Huang2013Near} designed efficient spectrum allocation algorithms to achieve truthfulness and maximize social welfare.  Unfortunately, none of these mechanisms considered any guarantee of privacy.

\textbf{Privacy-preserving spectrum auctions.} There are mainly two kinds of related works for privacy-preserving auctions. One is to protect privacy with cryptographical techniques from attacks inferring sensitive information by analyzing the process of auction computations, and the other is to protect privacy with the notion of differential privacy from attacks inferring sensitive information by analyzing the auction outcome. A brief review for these two kinds of works is given as follows.

\emph{Cryptographically private spectrum auctions.} Pan $et$ $al.$ proposed THEMIS \cite{Pan2011Purging}, the first secure spectrum auction to prevent malicious behaviors of the auctioneer. Later, Huang et al. \cite{Huang2013SPRING} presented SPRING, the first strategy-proof and privacy-preserving spectrum auction mechanism. Whereafter, a strategy-proof mechanism was designed to maximize social welfare without disclosing the bid privacy in \cite{Huang2015PPS}. Wu et al. \cite{wu2015towards} designed both strategy-proof and privacy-preserving auction mechanisms for noncooperative wireless networks. These works above mainly focused on privacy preservation for single-sided spectrum auctions. For double spectrum auctions, papers \cite{Chen2014PS,Chen2017Secure} employed various cryptography techniques to achieve privacy preservation. However, these studies preserved privacy in the process of auction computations, and they did not consider the privacy issues that resulted from analyzing the auction outcome.


\emph{Differentially private spectrum auctions.} Lately, a new concept of differential privacy \cite{Dwork2008Differential,Dwork2006Differential} is proposed. Then, McSherry and Talwar \cite{Mcsherry2007Mechanism} first combined differential privacy with mechanism design, paying attention to approximate truthfulness, collusion resilience and mechanism repeatability. Huang and Kannan \cite{huang2012exponential} further studied the exponential mechanism, the first general tool used to design mechanisms achieving both truthfulness and differential privacy. Xu et al. \cite{Xu2017PADS} first addressed the privacy-preserving cloud resource allocation problem based on differential privacy. Dong et al. \cite{dong2018protecting} designed a novel scheme to select spectrum-sharing secondary user in a differentially private manner. In \cite{Zhu2014Differentially,Zhu2015Differentially,Wu2016Designing}, applying differential privacy, authors designed mechanisms with approximate truthfulness and revenue maximization for single-sided spectrum auctions. As far as we know, none of the existing works protects differential privacy in the setting of double spectrum auctions.

\section{Technical Preliminaries}
\label{sec:preliminaries}
In this section, we provide the problem model of the double spectrum auction, and introduce some related solution concepts about mechanism design and differential privacy.

\subsection{Problem Model}
We consider a sealed-bid double spectrum auction with one auctioneer, $M$ sellers and $N$ buyers. We assume that each seller sells one distinct channel and each buyer buys one channel like the work \cite{yao2011efficient}. The channels are homogeneous to buyers, so that each buyer could request any channel. A sealed-bid auction is run by the auctioneer who is trustworthy. Bidders submit their bids privately to the auctioneer without any knowledge of others, then the auction mechanism outputs the final prices to determine the winners.

Spectrum auctions are different from traditional auctions, because spectrum has reusability of time and space (Here we mainly consider spatial reusability). Spatial reusability indicates that multiple buyers can share the same channel if they do not interfere with each other. In the auction, we want to protect the bidders' privacy and meanwhile maximize the social welfare. Here, the social welfare the auction mechanism is defined as the sum of the true values of all winning buyers subtracting that of all winning sellers. Better social welfare means more social value created by an auction.

In the auction, let $\mathbf{q}$ denote the profile of all quotations of sellers and $\mathbf{b}$ denote the profile of all bids of buyers. Each seller $m$ submits its quotation $q_m \in [1,q_{max}]$ for selling a channel, based on its true valuation $v_m^s$, and gets $p_m^s$ as its final selling price of the channel. Each buyer $n$ submits its bid $b_n \in [1,b_{max}]$ for buying a channel, based on its true valuation $v_n^b$, and gets $p_n^b$ as its final buying price of the channel. If the seller or buyer wins, it has the utility of $u_m=p_m^s-v_m^s$ or $u_n=p_n^b-v_n^b$; else the utility is 0. In a truthful auction, any bidder cannot improve its utility by submitting a dishonest quotation or bid, that is $q_m \neq v_m^s$ or $b_n \neq v_n^b$. When the auction is of individual rationality, each winning seller gets a selling price not lower than its quotation, and each winning buyer gets a buying price not higher than its bid.


\subsection{Related Solution Concepts}
we now review some important and related solution concepts used in this paper from mechanism design and differential privacy.

Firstly, we introduce some important definitions in mechanism design.
\begin{mydef}[Dominant Strategy \cite{Nisan2007Algorithmic}]
Strategy $s_i$ is a player $i$'s dominant strategy in a game, if for any strategy ${s'_i}$ $\neq$ $s_i$ and any other players' strategy profile $s_{-i}$,
\begin{displaymath}
u_i(s_i, s_{-i}) \geq u_i({s'_i}, s_{-i})
\end{displaymath}
\end{mydef}
The concept of truthfulness is related to dominant strategy. For each bidder in an auction, truthfulness means revealing its truthful value is a dominant strategy. However, exact truthfulness for  auctions is often too restrictive to be compatible with other properties such as social welfare maximization, and thus we turn to consider a weaker notion of $\gamma$-truthfulness. Moreover, for double spectrum auctions, we should  achieve both seller and buyer truthfulness. Therefore, according to the definition of approximate truthfulness in literature \cite{Gupta2010Differentially}, we give the definition of $\gamma$-truthfulness for double spectrum auctions as shown in Definition \ref{def:gamma-truthfulness}.
\begin{mydef}[$\gamma$-truthfulness]\label{def:gamma-truthfulness}
 Let $\mathbf{b}$ denote the profile of all buyers' truthful bids, $\mathbf{q}$ denote the profile of all sellers' truthful quotations. Let $b_n$ and $q_m$ denote the truthful strategies for any buyer $n$ and any seller $m$, respectively. A double auction mechanism is $\gamma$-truthfulness in expectation, if and only if for any strategy $b_n'$ $\neq$ $b_n$ it satisfies
\begin{displaymath}
E[u_n(b_n, \mathbf{q}, \mathbf{b_{-n}})] \geq E[u_n(b_n', \mathbf{q}, \mathbf{b_{-n}})]-{\gamma}
\end{displaymath}
at the same time, for any strategy $q_m'$ $\neq$ $q_m$ it satisfies
\begin{displaymath}
E[u_m(q_m, \mathbf{b}, \mathbf{q_{-m}})] \geq E[u_n(q_m', \mathbf{b}, \mathbf{q_{-m}})]-{\gamma}
\end{displaymath}
where $\mathbf{b_{-n}}$ represents the strategy profile of buyers except buyer $n$, $\mathbf{q_{-m}}$ represents the strategy profile of sellers except seller $m$, and $\gamma$ are a small constant.

\end{mydef}

Differential privacy has been widely studied in the field of theoretical computer science \cite{Dwork2008Differential}. This privacy concept ensures that for any two adjacent databases, the two corresponding outputs resulted are basically consistent. The specific definition in our context is as follows.
\begin{mydef}[Differential Privacy]
Let $\mathbf{b}$ denote the profile of all buyers' bids, $\mathbf{q}$ denote the profile of all sellers' quotations, and $\mathbf{D} = \mathbf{b} \cup \mathbf{q}$ represent the union profile of the two aforementioned profiles. A randomized double spectrum auction mechanism $\mathcal{M}$ gives $\epsilon$-differential privacy, if for any two profiles $\mathbf{D}$ and $\mathbf{D'}$ differing in the value of only one bid or quotation, and all $S$ $\subseteq$ Rang($\mathcal{M}$), it satisfies
\begin{displaymath}
Pr[\mathcal{M}(\mathbf{D})\in S]\leq exp(\epsilon)\times Pr(\mathcal{M}(\mathbf{D'})\in S)
\end{displaymath}
where $\epsilon$ is a small constant called the privacy budget representing the privacy level achieved.
\end{mydef}

A power tool to achieve differential privacy is exponential mechanism. It assigns a selection probability to each possible outcome according to a utility function that maps a pair of input and outcome to a utility score. It is defined as follows:
\begin{mydef}[Exponential Mechanism \cite{Mcsherry2007Mechanism}]
Given a range $\mathbb{P}$, and a utility function $F(\mathbf{D}, p)$ that takes as input a data profile $\mathbf{D}$ and an output p in the range $\mathbb{P}$, the exponential mechanism $\mathcal{M}(\mathbf{D}, F, \mathbb{P})$ outputs p from the range $\mathbb{P}$ with probability
\begin{displaymath}
Pr[\mathcal{M}(\mathbf{D}, F, \mathbb{P})=p]\propto exp(\epsilon F(\mathbf{D}, p)/2\Delta F)
\end{displaymath}
where $\Delta F$ is the sensitivity of function $F$, maximum difference between the two utility scores of any pair of data profiles $\mathbf{D}$ and $\mathbf{D'}$ differing in the value of a single element for any output p, and $\epsilon$ is the privacy budget.
\end{mydef}
Some theoretical results on differential privacy, which will be used in our later proofs, are depicted in the following lemmas.
\begin{mylemma}[Composition \cite{dwork2014algorithmic}]\label{the:composition}
Given the randomized algorithm $\mathcal{M}_1(D)$ that satisfies $\epsilon_1$-differential privacy and $\mathcal{M}_2(D)$ that satisfies $\epsilon_2$-differential privacy, then $\mathcal{M}(D)$=($\mathcal{M}_1(D)$, $\mathcal{M}_2(D))$ satisfies $(\epsilon_1+\epsilon_2)$-differential privacy.
\end{mylemma}
\begin{mylemma}[Convexity \cite{Li2016Differential}]\label{the:convexity}
Given $k$ mechanisms $\mathcal{M}_1$, $\mathcal{M}_2$, ..., $\mathcal{M}_k$ that satisfy $\epsilon$-differential privacy, and $p_1$, $p_2$, ..., $p_k$ $\in$ [0,1] such that $\sum_{i=1}^k{p_i}=1$, let $\mathcal{M}$ denote the mechanism that applies $\mathcal{M}_i$ with probability $p_i$. Then $\mathcal{M}$ satisfies $\epsilon$-differential privacy.
\end{mylemma}

\section{Our Mechanism: DDSM}\label{sec:ddsm-details}
In this section, we design our differentially private double spectrum auction mechanism, DDSM, and analyze its related properties.
\subsection{Design Rationale}


We integrate the exponential mechanism with double spectrum auction to achieve differential privacy, $\gamma$-truthfulness and approximate social welfare maximization. Basically, DDSM is a uniform price auction, where all sellers are paid with a selling clearing price, and all buyers are charged with a buying clearing price. Our main idea is to randomly select a selling clearing price from the set of selling prices, and then a buying clearing price from the set of buying prices, both with carefully designed probability distributions. These probability distributions should ensure that the auction mechanism produces outcome with guarantees of differential privacy, $\gamma$-truthfulness and approximate social welfare maximization. There are  mainly two design challenges, and we address them as follows.

The first challenge is how to perform the double-sided payment selection, while guaranteeing the differential privacy, $\gamma$-truthfulness and approximate social welfare maximization. We address this challenge by applying one exponential mechanism to select the selling clearing price and the other exponential mechanism to select the buying clearing price, and then sequentially compose the two mechanisms to get the whole design.


The second challenge is how to group the buyers while preserving differential privacy. To address this challenge, our main idea is to use a deterministic buyer grouping algorithm independent of buyers' bids, generating a fix grouping for any set of buyers. In next section, we will show that a randomized grouping algorithm independent of buyers' bids can be also applicable.

\subsection{Design Detail}
 Now, we describe the detailed design of DDSM. Specifically, DDSM consists of the following three steps.

\textbf{(1) Buyer Group Formation}

In order to exploit spatial reusability of spectrum, in this step DDSM gathers buyers into groups, such that members in the same group do not conflict with each other. Moreover, the group formation should be deterministic and independent of buyers' bids, and hence changing any buyer's bid will not affect the group formation. That is, by defining neighboring bid profiles as those only differing in a single bid value, we can see that any two neighboring bid profiles respond to the same group form. Specifically, we let the auctioneer construct a conflict graph of buyers according to their geographic locations, and then deterministically find independent sets of nodes (buyers) as the buyer groups.  Buyers in the same group can use the same channel. We denote the buyer groups as $G=\{G_l|l=1, 2, ..., L\}$.


\textbf{(2) Double-sided Price Selection}

In this step, DDSM employs one exponential mechanism to select the selling clearing price and the other to select the buying clearing price, both with related social welfare values as their utility functions. We describe the details in the following.

First, the sets of both quotations and group bids are computed. The set of quotations can be simply defined as $[1 .. q_{\max}]$, where $q_{\max}$ denotes the maximum of all possible quotations. To compute the set of group bids, we need to compute the maximum of all possible group bids. The group bid $g_l$ of group $G_l$ is defined as
\begin{displaymath}
g_l=\min_{1 \le i \le |G_l|}{b_l(i)} \times |G_l|\\
\end{displaymath}
where $b_l(i)$ denotes the $i$th buyer's bid in group $G_l$. Let the maximum of all possible bids be $b_{\max}$, and the maximum size of all buyer groups be $n_{\max}$, we get the maximum group bid is $n_{\max} \cdot b_{\max}$, and thus the set of group bids is $[1 .. n_{\max} \cdot b_{\max}]$.

Second, for fairness and efficiency, DDSM naturally sorts sellers' quotations in non-decreasing order and buyer groups' bids in non-increasing order.
\begin{displaymath}
\mathbb{O'}:q_1\leq q_2\leq \ldots \leq q_M\\
\end{displaymath}
\begin{displaymath}
\mathbb{O''}:g_1\geq g_2\geq \ldots \geq g_L\\
\end{displaymath}

That is, sellers with lower quotations are made to sell their channels first, while groups with higher bids are made to buy their channels first, too.

Third, the selling clearing price $p_s$ and then the buying clearing price $p_g$ are randomly selected with their respective appropriate probability distributions. The two probability distributions should be properly related to the social welfare values that both clearing prices cause. Furthermore, to make all resulted transactions are profitable, it is necessary that $p_s \le p_g$. In the following, we first describe how to compute the related social welfare values given a pair of clearing prices $p_s$ and $p_g$, and then account for how to compute the appropriate probability distributions.

\textit{Computation of social welfare values}. Given a pair of clearing prices $p_s \in [1..q_{max}]$ and $p_g \in [p_s..n_{max} \cdot b_{max}]$, we can find the set of $k_s = \arg \max_m \{q_m \le p_s\}$ potential winning sellers in $\mathbb{O'}$, and the set of $k_g = \arg \max_l \{g_l \ge p_g\}$ potential winning buyer groups in $\mathbb{O''}$. We denote the two sets by TOP $(\mathbb{O'}, k_s)$ and TOP $(\mathbb{O''}, k_g)$, respectively. Obviously, these potential winners can make up at most $k = \min(k_s, k_g)$ profitable transactions. Thus, we can determine $k$ potential winning sellers and $k$ potential winning buyer groups as the winners, and compute the resulted social welfare values. Naturally, if $k_s = k$ (resp. $k_g = k$), then all potential winning sellers (resp. buyer groups) are the winners. However, if $k_s > k$ (resp. $k_g > k$), how to determine $k$ winners out of more than $k$ potential winners is a subtlety for achieving truthfulness. To prevent bid (resp. quotation) operation, the winner determination should be independent on bids (resp. quotations). Especially, note that selecting the $k$ potential winners with the $k$th highest bids (or quotations) allows market operations. So, if $k_s > k$ (resp. $k_g > k$), we
just randomly select $k$ out of $k_s$ (resp. $k_g$) potential winners as the winners.

We represent the sets of randomly selected $k$ winning sellers and $k$ winning buyer groups as follows
\[\mathbb{W}^s = \text{RAND}(\text{TOP}(\mathbb{O'}, k_s), k)\]
\[\mathbb{W}^g = \text{RAND}(\text{TOP}(\mathbb{O''}, k_g), k)\]

Then the $l$th seller in $\mathbb{W}^s$ and the $l$th buyer group in $\mathbb{W}^g$ form the $l$th transaction for winners. The social welfare produced by the $l$th transaction can be computed as
\begin{displaymath}
w_l=\sum _{i=1}^{|G'_l|}{b_l(i)} - q'_l
\end{displaymath}
where $G'_l$ represents the $l$th buyer group in $\mathbb{W}^g$, and $q'_l$ represents the quotation of the $l$th seller in $\mathbb{W}^s$.

Therefore, given any pair of clearing prices $p_s$ and $p_g$, we can compute the resulted social welfare value as follows.
\begin{equation}\label{eq:social-welfare}
W(p_s, p_g) = \sum_{1 \le l \le k}{w_l}
\end{equation}

Note that we write the resulted social welfare as $W(p_s, p_g)$ to explicitly indicate its dependence on $p_s$ and $p_g$.

\textit{Computation of probability distributions}. To achieve differential privacy, we need to select the selling clearing price $p_s$ and then the buying clearing price $p_g$ with appropriate probability distributions. To do this, one exponential mechanism is applied for each clearing price selection.

To select $p_s$ with the exponential mechanism, we need to design a utility function that is only dependent on $p_s$ and reflects a better social welfare value with a higher utility. We let $\mathcal{M}_1$ denote the mechanism selecting $p_s$, and define the utility function as follows.
\begin{equation}
 W_1(\mathbf{q}, \mathbf{b}, p_s) = \max_{p_g} W(p_s, p_g)
\end{equation}
where we also indicate explicitly the dependence of $W_1$ on the buyers' bid profile $\mathbf{b}$ and the sellers' quotation profile $\mathbf{q}$. Since changing any bid or quotation has a maximum impact of ${n_{max}}\cdot{b_{max}} - 1$ on the social welfare, we get the sensitivity $\Delta W_1={n_{max}}\cdot{b_{max}}-1$.

According to the exponential mechanism, the probability distribution of the selling clearing price can be computed as follows.
\begin{equation}\label{eq:ps-distr}
Pr(\mathcal{M}_1(\mathbf{q}, \mathbf{b}) = p_s)=\frac{{\exp}(\frac{\epsilon_1 W_1(\mathbf{q}, \mathbf{b}, p_s)}{2\Delta W_1})}{\sum_{p_s'\in\mathbb{P}_s}{\exp}(\frac{\epsilon_1 W_1(\mathbf{q}, \mathbf{b}, p_s')}{2\Delta W_1})}
\end{equation}
for all $p_s \in \mathbb{P}_s$, where $\mathbb{P}_s = [1..q_{max}]$, and $\epsilon_1$ is the privacy budget.

For selecting $p_g$ given $p_s$ with the exponential mechanism, we let $\mathcal{M}_2$ denote this mechanism, and similarly define the utility function as follows.
\begin{equation}
 W_2(\mathbf{q}, \mathbf{b}, p_s, p_g) = W(p_s, p_g)
\end{equation}
where we have $\Delta W_2={n_{max}}\cdot{b_{max}}-1$.

Then, according to the exponential mechanism, the probability distribution of the buying clearing price can be computed as follows.
\begin{equation}\label{eq:pg-distr}
Pr(\mathcal{M}_2(\mathbf{q}, \mathbf{b}, p_s)=p_g)=\frac{{\exp}(\frac{\epsilon_2 W_2(\mathbf{q}, \mathbf{b}, p_s, p_g)}{2\Delta W_2})}{\sum_{p_g'\in\mathbb{P}_g(p_s)}{\exp}(\frac{\epsilon_2 W_2(\mathbf{q}, \mathbf{b}, p_s, p_g')}{2\Delta W_2})}
\end{equation}
for all $p_g \in \mathbb{P}_g(p_s)$. Here, we let $\mathbb{P}_g(p_s) = [p_s..n_{max} \cdot q_{max}]$, and $\epsilon_2$ is the privacy budget.

To sum up, this step randomly selects the selling and buying clearing prices $p_s$ and $p_g$ in term of the probability distributions \eqref{eq:ps-distr} and \eqref{eq:pg-distr}, respectively. The detailed procedure is depicted in Algorithm~\ref{protocol:price-selection}.

\renewcommand{\algorithmicrequire}{\textbf{Input:}}  
\renewcommand{\algorithmicensure}{\textbf{Output:}}  

\begin{algorithm}[htb]
\caption{Double-sided Price Selection} \label{protocol:price-selection}
\begin{algorithmic}[1]

\REQUIRE The set of sellers $\mathbb{M}$, the sellers' quotation profile $\mathbf{q}$, the set of buyers $\mathbb{N}$, the buyers' bid profile $\mathbf{b}$, the buyer grouping $G$, and privacy budgets $\epsilon_1$ and $\epsilon_2$ with $\epsilon_1 + \epsilon_2 = \epsilon$.
\ENSURE Selling and buying clearing prices $p_s$ and $p_g$

\vspace{0.5\baselineskip}

(1) \underline{Compute group bids:}

\FOR{$l \leftarrow 1$ to $L$}

\STATE $g_l \leftarrow \min_{1 \le i \le |G_l|}{b_l(i)} \times |G_l|$

\ENDFOR

(2) \underline{Sort quotations and bids:}

\STATE $\mathbb{O'}:q_1\leq q_2\leq \ldots \leq q_M$

\STATE $\mathbb{O''}:g_1\geq g_2\geq \ldots \geq g_L$

//\emph{For simplicity, same indexes are abused after sorting}

(3) \underline{Compute social welfare values:}

\FOR{$p_s \leftarrow 1$ to $q_{max}$}

\FOR{$p_g \leftarrow p_s$ to $n_{max}\cdot b_{max}$}

\STATE  $k_s \leftarrow \arg \max_m \{q_m \le p_s\}$
\STATE  $k_g \leftarrow \arg \max_l \{g_l \ge p_g\}$
\STATE $k \leftarrow \min(k_s, k_g)$
\STATE $\mathbb{W}^s \leftarrow \text{RAND}(\text{TOP}(\mathbb{O'}, k_s), k)$
\STATE $\mathbb{W}^g \leftarrow \text{RAND}(\text{TOP}(\mathbb{O''}, k_g), k)$
\STATE Compute $W(p_s, p_g)$ from $\mathbb{W}^s$ and $\mathbb{W}^g$ using Eq.~\eqref{eq:social-welfare}

\ENDFOR

\ENDFOR

(4) \underline{Compute probability distributions and select prices:}

\FOR{$p_s \leftarrow 1$ to $q_{max}$}

\STATE $ W_1(\mathbf{q}, \mathbf{b}, p_s) \leftarrow \max_{p_g} W(p_s, p_g)$
\STATE $\Delta W_1 \leftarrow {n_{max}}\cdot{b_{max}} - 1$
\STATE $Pr(\mathcal{M}_1(\mathbf{q}, \mathbf{b}) = p_s) \leftarrow \frac{{\exp}(\frac{\epsilon_1 W_1(\mathbf{q}, \mathbf{b}, p_s)}{2\Delta W_1})}{\sum_{p_s'\in\mathbb{P}_s}{\exp}(\frac{\epsilon_1 W_1(\mathbf{q}, \mathbf{b}, p_s')}{2\Delta W_1})}$


\ENDFOR

\STATE $p_s \leftarrow \mathcal{M}_1(\mathbf{q}, \mathbf{b})$ // \emph{Select selling clearing price}

\FOR{$p_g \leftarrow p_s$ to $n_{max}\cdot b_{max}$}

\STATE $W_2(\mathbf{q}, \mathbf{b}, p_s, p_g) \leftarrow W(p_s, p_g)$

\STATE $\Delta W_2 \leftarrow {n_{max}}\cdot{b_{max}} - 1$

\STATE $Pr(\mathcal{M}_2(\mathbf{q}, \mathbf{b}, p_s)=p_g) \leftarrow $ \\$\frac{{\exp}(\frac{\epsilon_2 W_2(\mathbf{q}, \mathbf{b}, p_s, p_g)}{2\Delta W_2})}{\sum_{p_g'\in\mathbb{P}_g(p_s)}{\exp}(\frac{\epsilon_2 W_2(\mathbf{q}, \mathbf{b}, p_s, p_g')}{2\Delta W_2})}$

\ENDFOR

\STATE $p_g \leftarrow \mathcal{M}_2(\mathbf{q}, \mathbf{b}, p_s)$  // \emph{Select buying clearing price}

\RETURN $p_s$ and $p_g$

\end{algorithmic}
\end{algorithm}

\textbf{(3) Auction Outcome Release}

Having determined the clearing prices $p_s$ and $p_g$, we now turn to compute the resulted auction outcome. Just using exactly the same method and the same randomness when we compute the social welfare $W(p_s, p_g)$ given $p_s$ and $p_g$ in the previous step, we can easily find all winning sellers and all winning buyer groups. And then for each winning buyer in buyer group $G_l$, it will pay $p_g/|G_l|$. Finally, the resulted social welfare is just $W(p_s, p_g)$.

\subsection{Analysis}
Now, We prove that DDSM achieves differential privacy, $\gamma$-truthfulness, individual rationality, budget balance and approximate social welfare maximization.

Theorem \ref{the:dp} states the $\epsilon$-differential privacy of DDSM.

\begin{mytheorem}\label{the:dp}
DDSM achieves $\epsilon$-differential privacy, where $\epsilon=\epsilon_1+\epsilon_2$.
\end{mytheorem}

\begin{proof}
We denote a value profile $\mathbf{B} = \mathbf{q} \cup \mathbf{b}$ as the union of any pair of quotation and bid profiles $(\mathbf{q}, \mathbf{b})$, and denote any neighboring profile by $\mathbf{B'}$, which is only differing in a single bid or quotation value from $\mathbf{B}$. Let $\mathcal{M}_1$ and $\mathcal{M}_2$ denote the two mechanisms to randomly select $p_s$ and $p_g$, respectively. The probability ratio of selecting $p_s \in \mathbb{P}_s$ as the selling clearing price for any pair of profiles $\mathbf{B}$ and $\mathbf{B'}$ is
\begin{displaymath}
\begin{aligned}
&\frac{Pr(\mathcal{M}_1(\mathbf{B})=p_s)}{Pr(\mathcal{M}_1(\mathbf{B'})=p_s)}\\
=&\frac{\frac{{\exp}(\frac{\epsilon_1 W_1(\mathbf{B}, p_s)}{2\Delta W_1})}{\sum_{p_s'\in\mathbb{P}_s}{\exp}(\frac{\epsilon_1 W_1(\mathbf{B}, p_s')}{2\Delta W_1})}}{\frac{{\exp}(\frac{\epsilon_1 W_1(\mathbf{B'}, p_s)}{2\Delta W_1})}{\sum_{p_s'\in\mathbb{P}_s}{\exp}(\frac{\epsilon_1 W_1(\mathbf{B'}, p_s')}{2\Delta W_1})}}\\
=&(\frac{{\exp}(\frac{\epsilon_1 W_1(\mathbf{B}, p_s)}{2\Delta W_1})}{{\exp}(\frac{\epsilon_1 W_1(\mathbf{B'}, p_s)}{2\Delta W_1})})(\frac{{\sum_{p_s'\in\mathbb{P}_s}{\exp}(\frac{\epsilon_1 W_1(\mathbf{B'}, p_s')}{2\Delta W_1})}}{{\sum_{p_s'\in\mathbb{P}_s}{\exp}(\frac{\epsilon_1 W_1(\mathbf{B}, p_s')}{2\Delta W_1})}})\\
\leq& \exp(\frac{\epsilon_1}{2})(\frac{{\sum_{p_s'\in\mathbb{P}_s}\exp(\frac{\epsilon_1}{2}){\exp}(\frac{\epsilon_1 W_1(\mathbf{B}, p_s')}{2\Delta W_1})}}{{\sum_{p_s'\in\mathbb{P}_s}{\exp}(\frac{\epsilon_1 W_1(\mathbf{B}, p_s')}{2\Delta W_1})}})\\
\leq & \exp(\frac{\epsilon_1}{2})\exp(\frac{\epsilon_1}{2})(\frac{{\sum_{p_s'\in\mathbb{P}_s}{\exp}(\frac{\epsilon_1 W_1(\mathbf{B}, p_s')}{2\Delta W_1})}}{{\sum_{p_s'\in\mathbb{P}_s}{\exp}(\frac{\epsilon_1 W_1(\mathbf{B}, p_s')}{2\Delta W_1})}})\\
=& \exp(\epsilon_1)
\end{aligned}
\end{displaymath}

Symmetrically, we have $\exp(-\epsilon_1) \le \frac{Pr(\mathcal{M}_1(\mathbf{B})=p_s)}{Pr(\mathcal{M}_1(\mathbf{B'})=p_s)}$. Thus, the $\mathcal{M}_1$ achieves $\epsilon_1$-differential privacy.

Similarly, it can be proved that $M_2$ achieves $\epsilon_2$-differential privacy. Finally, by the composition lemma(Lemma \ref{the:composition}), DDSM achieves $\epsilon$-differential privacy, with $\epsilon = \epsilon_1 + \epsilon_2$. 
\end{proof}

The $\gamma$-truthfulness of the DDSM is stated in Theorem \ref{the:gamma-truthfulness}.

\begin{mytheorem}\label{the:gamma-truthfulness}
DDSM is $\gamma$-truthful, where $\gamma$ = max($\epsilon_1u_{1max}$, $\epsilon_2u_{2max}$).
\end{mytheorem}

\begin{proof}
For double spectrum auctions, we need to prove the truthfulness for both sellers and buyers. We describe the proof in the two aspects as follows.

\emph{(1) Proof of $\epsilon_1u_{1max}$-truthfulness for sellers}

For sellers, the ranges of both prices and true values are $[1, q_{max}]$, so a seller's maximum utility is $u_{1max}=q_{max}-1$. We let $w_m$ denote whether a seller $m$ wins ($w_m = 1$) or not ($w_m = 0$). It is worth noting that in the winner determination, we take a random selection of $k$ sellers when the potential winning sellers are more than $k$. So no matter how a seller changes its own quotation within the range of less than the clearing price $p_s$, it has no effect on the auction outcome, $i.e, w_m(q_m, \mathbf{q_{-m}}, \mathbf{b}, p_s, p_g)=w_m(v_m^s, \mathbf{q_{-m}}, \mathbf{b}, p_s, p_g)$. We analyze the expected  utility of any seller $m$ in two cases as follows.
\begin{itemize}
    \item Case\ $q_m \geq v_m^s$: For any $p_s\in\mathbb{P}_s$, if $q_m>p_s$, the seller $m$ does not win ($i. e, w_m=0$) and his utility is 0; else $q_m\leq p_s$, the seller $m$ is randomly selected as a winner no matter its quotation is $q_m$ or $v_m^s$. Then, if the seller $m$ wins($i. e, w_m=1$), its utility is $u_m=(p_s-v_m^s)>0$, else the seller $m$ is lost($i.e, w_m=0$), and his utility is 0. So we have
        \begin{displaymath}\small
        \begin{aligned}
        \qquad &E[u_m(v_m^s, \mathbf{q_{-m}}, \mathbf{b}, \mathbb{P}_s)]\\
        =\quad&\sum_{p_s\in\mathbb{P}_s}Pr[\mathcal{M}_1(v_m^s, \mathbf{q_{-m}}, \mathbf{b})=p_s]\\
        \quad&\times w_m\cdot(p_s-v_m^s)\\
        \geq \quad &exp(-\epsilon_1)\sum_{p_s\in\mathbb{P}_s}Pr[\mathcal{M}_1(q_m, \mathbf{q_{-m}}, \mathbf{b})=p_s]\\
        \quad&\times w_m\cdot(p_s-v_m^s)\\
        =\quad&exp(-\epsilon_1)E[u_m(q_m, \mathbf{q_{-m}}, \mathbf{b}, \mathbb{P}_s)]\\
        \geq \quad &(1-\epsilon_1)E[u_m(q_m, \mathbf{q_{-m}}, \mathbf{b}, \mathbb{P}_s)]\\
        =\quad&E[u_m(q_m, \mathbf{q_{-m}}, \mathbf{b}, \mathbb{P}_s)]-\epsilon_1E[u_m(q_m, \mathbf{q_{-m}}, \mathbf{b}, \mathbb{P}_s)]\\
        \geq \quad&E[u_m(q_m, \mathbf{q_{-m}}, \mathbf{b}, \mathbb{P}_s)]-\epsilon_1u_{1max}\\
        \end{aligned}
        \end{displaymath}

    \item Case\ $q_m < v_m^s$: For any $p_s\in\mathbb{P}_s$, if $v_m^s < p_s$ and $q_m < v_m^s$, then the seller $m$ wins with a certain probability ($i. e, $ $w_m=1$) and his utility is $u_m=(p_s-v_m^s)>0$; else if $v_m^s > p_s > q_m$ and the seller $m$ is randomly selected, then the seller $m$ also wins($i. e, $ $w_m=1$), but his utility is $u_m=(p_s-v_m^s)<0$ or the seller $m$ is not a winner, his utility is 0 ; else $v_m^s > q_m >p_s$, then the seller $m$ does not win($i. e, $ $w_m=0$) and his utility is 0. So
        \begin{displaymath}\small
        \begin{aligned}
        \qquad &E[u_m(v_m^s, \mathbf{q_{-m}}, \mathbf{b}, \mathbb{P}_s)]\\
        =\quad&\sum_{p_s\in\mathbb{P}_s}Pr[\mathcal{M}_1(v_m^s, \mathbf{q_{-m}}, \mathbf{b})=p_s]\times w_m\cdot(p_s-v_m^s)\\
        \geq \quad &\sum_{(p_s\in\mathbb{P}_s)\bigwedge(p_s>v_m^s)}Pr[\mathcal{M}_1(v_m^s, \mathbf{q_{-m}}, \mathbf{b})=p_s]\\
        \quad&\times w_m\cdot(p_s-v_m^s)\\
        = \quad &\sum_{(p_s \in \mathbb{P}_s)\bigwedge(p_s>v_m^s)}Pr[\mathcal{M}_1(q_m, \mathbf{q_{-m}}, \mathbf{b})=p_s]\\
        \quad&\times w_m\cdot(p_s-v_m^s)\\
        \geq \quad &E[u_m(q_m, \mathbf{q_{-m}}, \mathbf{b},\mathbb{P}_s)]\\
        \end{aligned}
        \end{displaymath}
 \end{itemize}

 From the above, for any seller $m$, we have
\begin{equation}
 E[u_m(v_m^s, \mathbf{q_{-m}}, \mathbf{b}, \mathbb{P}_s)]\geq E[u_m(q_m, \mathbf{q_{-m}}, \mathbf{b}, \mathbb{P}_s)]-\epsilon_1u_{1max}
 \end{equation}

Thus, DDSM is $\epsilon_1u_{1max}$-truthful for sellers.\\

\emph{(2) Proof of $\epsilon_2u_{2max}$-truthfulness for buyers}

For buyers, the ranges of both prices and true values are $[1, b_{max}]$, so a buyer's maximum utility is $u_{2max}=b_{max}-1$. Employing mechanism $\mathcal{M}_2$, we can determine the buying clearing price $p_g$ for buyer groups. Note that mechanism $\mathcal{M}_2$ selecting the buying clearing price has obtained the seller clearing price expressed by $p_s^1$. Then the buying clearing price for buyers in each winning buyer group $G_l$, $p_l^b=\frac{p_g}{|G_l|}$. Similarly, let $w_n$ denote whether a buyer $n$ wins, and we have $w_n(b_n, \mathbf{q}, \mathbf{b_{-n}}, p_n^b)=w_n(v_n^b, \mathbf{q}, \mathbf{b_{-n}}, p_n^b)$ once buyer $n$ is in a potential winning buyer group due to the random selection of winning groups. For any buyer $n$, let $b_{min}$, $g_l$ denote the minimum bid and group bid, respectively, of the group $G_l$ that buyer $n$ belongs to, and we analyze buyer $n$'s utility in the following two cases.

\begin{itemize}
 \item Case\ $b_n\leq v_n^b$: For any $p_g\in\mathbb{P}_g(p_s^1)$, if $g_l < p_g$ or $g_l \geq p_g$, and the group $G_l$ is lost ($i. e, $ $w_n=0$), the utility of the buyer $n$ is 0; else $g_l \geq p_g$ and the group $l$ is randomly selected, then $b_n\geq b_{min}>p_n^b$ and $b_n\leq v_n^b$, the group of the buyer $n$ wins($i. e, w_n=1$)and his utility is $u_n=(v_n^b-p_n^b)>0$. So
     \begin{displaymath}\small
     \begin{aligned}
    \qquad &E[u_n(v_n^b, \mathbf{q}, \mathbf{b_{-n}}, p_s^1, \mathbb{P}_g(p_s^1))]\\
    =\quad&\sum_{p_g\in\mathbb{P}_g(p_s^1)}Pr[\mathcal{M}_2(v_n^b, \mathbf{q}, \mathbf{b_{-n}})=p_g]\times w_n\cdot(v_n^b-p_n^b)\\
    \geq \quad &exp(-\epsilon_2)\sum_{p_g\in\mathbb{P}_g(p_s^1)}Pr[\mathcal{M}_2(b_n, \mathbf{q}, \mathbf{b_{-n}})=p_g]\\
    \quad&\times w_n\cdot(v_n^b-p_n^b)\\
    =\quad&exp(-\epsilon_2)E[u_n(b_n, \mathbf{q}, \mathbf{b_{-n}}, p_s^1, \mathbb{P}_g(p_s^1))]\\
    \geq \quad &(1-\epsilon_2)E[u_n(b_n, \mathbf{q}, \mathbf{b_{-n}}, p_s^1, \mathbb{P}_g(p_s^1))]\\
    =\quad&E[u_n(b_n, \mathbf{q}, \mathbf{b_{-n}}, p_s^1, \mathbb{P}_g(p_s^1))]\\
    \quad&-\epsilon_2E[u_n(b_n, \mathbf{q}, \mathbf{b_{-n}}, p_s^1, \mathbb{P}_g(p_s^1))]\\
    \geq \quad&E[u_n(b_n, \mathbf{q}, \mathbf{b_{-n}}, p_s^1, \mathbb{P}_g(p_s^1))]-\epsilon_2u_{2max}\\
    \end{aligned}
    \end{displaymath}

\item Case\ $b_n>v_n^b$: For any $p_g\in\mathbb{P}_g(p_s^1)$, if $b_n<p_n^b$ or $b_n>p_n^b>b_{min}$, the buyer $n$ does not win($i. e, $ $w_n=0$)and his utility is 0; else if $g_l \geq p_g$ and the group $l$ is selected with a certain probability, $b_n\geq b_{min}>p_n^b$ and $v_n^b>p_n^b$, then the buyer $n$ wins($i. e, $ $w_n=1$) and his utility is $u_n=(v_n^b-p_n^b)>0$; if $b_n\geq b_{min}>p_n^b$ and $v_n^b<p_n^b$, then the buyer $n$ wins($i. e, $ $w_n=1$), but his utility is $u_n=(v_n^b-p_n^b)<0$. So
    \begin{displaymath}\small
    \begin{aligned}
    \qquad &E[u_n(v_n^b, \mathbf{q}, \mathbf{b_{-n}}, p_s^1, \mathbb{P}_g(p_s^1))]\\
    \geq \quad &\sum_{(p_g\in\mathbb{P}_g(p_s^1))\bigwedge(p_n^b<v_n^b)}Pr[\mathcal{M}_2(v_n^b, \mathbf{q}, \mathbf{b_{-n}})=p_g]\\
    \quad& \times w_n\cdot(v_n^b-p_n^b)\\
    = \quad &\sum_{(p_g\in\mathbb{P}_g(p_s^1))\bigwedge(p_n^b<v_n^b)}Pr[\mathcal{M}_2(b_n, \mathbf{q}, \mathbf{b_{-n}})=p_g]\\
    \quad& \times w_n\cdot(v_n^b-p_n^b)\\
    \geq \quad &E[u_n(b_n, \mathbf{q}, \mathbf{b_{-n}}, p_s^1, \mathbb{P}_g(p_s^1))]\\
    \end{aligned}
    \end{displaymath}
\end{itemize}

From the above, for any buyer $n$, we have
\begin{equation}
\begin{aligned}
\quad &E[u_n(v_n^b, \mathbf{q}, \mathbf{b_{-n}}, p_s^1, \mathbb{P}_g(p_s^1))]\\
\geq \quad& E[u_n(b_n, \mathbf{q}, \mathbf{b_{-n}}, p_s^1, \mathbb{P}_g(p_s^1))]-\epsilon_2u_{2max}
\end{aligned}
\end{equation}

Thus, DDSM is $\epsilon_2u_{2max}$-truthful for buyers.
Let $\gamma=\max(\epsilon_1u_{1max}, \epsilon_2u_{2max})$. We therefore have proved that DDSM is $\gamma$-truthful. 
\end{proof}

It can be also proved that DDSM approximately maximizes the social welfare as shown in Theorem \ref{the:social-welfare}.

\begin{mytheorem}\label{the:social-welfare}
DDSM achieves approximate social welfare maximization.
\end{mytheorem}

\begin{proof}
We first analyze the expected utilities resulted from mechanism $\mathcal{M}_1$ selecting selling clearing price $p_s$ and from mechanism $\mathcal{M}_2$ selecting buying clearing price $p_g$, and then analyze the expected utility (social welfare) of the overall mechanism of DDSM.

Let $\mathbf{B} = \mathbf{q} \cup \mathbf{b}$, and $C=n_{max}\cdot b_{max}-1$. Let $OPT_1$ = $max_{p\in\mathbb{P}_s}$ $W_1(\mathbf{B}, p)$ denote the maximum $W_1$ for mechanism $\mathcal{M}_1$. Define the sets $R_{t_1}=\{p \in \mathbb{P}_s:W_1(\mathbf{B}, p)>OPT_1-t_1\}$ and $\overline{R_{2t_1}}=\{p \in \mathbb{P}_s: W_1(\mathbf{B}, p)\leq OPT_1-2t_1\}$ for a small constant $t_1>0$.
We have
\begin{displaymath}
\begin{aligned}
Pr(\mathcal{M}_1(\mathbf{B}) \in \overline{R_{2t_1}})&\leq \frac{Pr(\mathcal{M}_1(\mathbf{B}) \in \overline{R_{2t_1}})}{Pr(\mathcal{M}_1(\mathbf{B}) \in R_{t_1})}\\
&=\frac{\sum_{p\in \overline{R_{2t_1}}}\frac{{exp}(\frac{\epsilon_1 W_1(\mathbf{B}, p)}{2C})}{\sum_{p'\in\mathbb{P}_s}{exp}(\frac{\epsilon_1 W_1(\mathbf{B}, p')}{2C})}}{\sum_{p\in R_{t_1}}\frac{{exp}(\frac{\epsilon_1 W_1(\mathbf{B}, p)}{2C})}{\sum_{p'\in\mathbb{P}_s}{exp}(\frac{\epsilon_1 W_1(\mathbf{B}, p')}{2C})}}\\
&=\frac{\sum_{p\in \overline{R_{2t_1}}}{{exp}(\frac{\epsilon_1 W_1(\mathbf{B}, p)}{2C})}}{\sum_{p\in R_{t_1}}{exp}(\frac{\epsilon_1 W_1(\mathbf{B}, p)}{2C})}\\
&\leq \frac{|\overline{R_{2t_1}}|{exp}(\frac{\epsilon_1(OPT_1-2t_1)}{2C})}{|R_{t_1}|{exp}(\frac{\epsilon_1(OPT_1-t_1)}{2C})}\\
&=\frac{|\overline{R_{2t_1}}|}{|R_{t_1}|}{exp}(\frac{-\epsilon_1t_1}{2C})\\
&\leq |\mathbb{P}_s|{exp}(\frac{-\epsilon_1t_1}{2C})\\
\end{aligned}
\end{displaymath}

Then, $Pr(\mathcal{M}_1(\mathbf{B})\in R_{2t_1})\geq 1-|\mathbb{P}_s|{exp}(\frac{-\epsilon_1t_1}{2C})$. If we let $t_1\geq \frac{2Cln(\frac{|\mathbb{P}_s|OPT_1}{t_1})}{\epsilon_1}$, we get $Pr(\mathcal{M}_1(\mathbf{B}) \in R_{2t_1})\geq 1-\frac{t_1}{OPT_1}$.

Therefore, for any $t_1\geq \frac{2Cln(\frac{|\mathbb{P}_s|OPT_1}{t_1})}{\epsilon_1}$, we have
\begin{equation}\small\label{eq:expected-w1}
\begin{aligned}
E_{p \in \mathbb{P}_s}[W_1(\mathbf{B}, p)]&>\sum_{p\in R_{2t_1}}W_1(\mathbf{B}, p)Pr(\mathcal{M}_1(\mathbf{B})=p)\\
&>(OPT_1-2t_1)(1-\frac{t_1}{OPT_1})\\
&>OPT_1-3t_1
\end{aligned}
\end{equation}

If we let $t_1=\frac{2Cln(e+\frac{\epsilon_1|\mathbb{P}_s|OPT_1}{2C})}{\epsilon_1}\geq \frac{2C}{\epsilon_1}$, and we have
\begin{displaymath}\small
\begin{aligned}
\frac{2Cln(\frac{|\mathbb{P}_s|OPT_1}{t_1})}{\epsilon_1}&\leq \frac{2Cln(e+\frac{\epsilon_1|\mathbb{P}_s|OPT_1}{2C})}{\epsilon_1}\\
&=t_1
\end{aligned}
\end{displaymath}

Then put $t_1$ into Eq.~\eqref{eq:expected-w1}, we get the expected utility of mechanism $\mathcal{M}_1$
\begin{displaymath}
\begin{aligned}
E_{p \in \mathbb{P}_s}[W_1(\mathbf{B}, p)]&>OPT_1-3t_1\\
&>OPT_1-3\frac{2ln(e+\frac{\epsilon_1|\mathbb{P}_s|OPT_1}{2C})}{\epsilon_1}
\end{aligned}
\end{displaymath}

Now we analyze the expected utility of mechanism $\mathcal{M}_2$ given a selling clearing price $p_s$. Let the optimal utility of $\mathcal{M}_2$ be $OPT_2=\max_{p\in \mathbb{P}_g(p_s)}W_2(\mathbf{B}, p_s, p)$. Similarly, define the sets $R_{t_2}=\{p \in \mathbb{P}_g(p_s): W_2(\mathbf{B}, p_s, p)>OPT_2-t_2\}$ and $\overline{R_{2t_2}}=\{p  \in \mathbb{P}_g(p_s): W_2(\mathbf{B}, p_s, p)\leq OPT_2-2t_2\}$ for a small constant $t_2>0$. Then, we have
\begin{displaymath}
Pr(\mathcal{M}_2(\mathbf{B}) \in \overline{R_{2t_2}})\leq |\mathbb{P}_g(p_s)|{exp}(\frac{-\epsilon_2t_2}{2C})
\end{displaymath}

For $t_2\geq \frac{2Cln(\frac{|\mathbb{P}_g(p_s)|OPT_2}{t_2})}{\epsilon_2}$,
\begin{displaymath}
E_{p \in \mathbb{P}_g(p_s)}[W_2(\mathbf{B}, p_s, p)] > OPT_2-3t_2
\end{displaymath}

Let $t_2=\frac{2Cln(e+\frac{\epsilon_2|\mathbb{P}_g(p_s)|OPT_2}{2C})}{\epsilon_2}\geq \frac{2C}{\epsilon_2}$, we get the expected utility of mechanism $\mathcal{M}_2$ given $p_s$
\begin{displaymath}
\begin{aligned}
E_{p \in \mathbb{P}_g(p_s)}[W_2(\mathbf{B}, p_s, p)] > OPT_2-3\frac{2ln(e+\frac{\epsilon_2|\mathbb{P}_g(p_s)|OPT_2}{2C})}{\epsilon_2}
\end{aligned}
\end{displaymath}

In term of the definition of $W_1$, we have $W_1(\mathbf{B}, p_s)=\max_{p\in \mathbb{P}_g(p_s)}W_2(\mathbf{B}, p_s, p)=OPT_2(p_s)$. Then, we obtain the expected utility (social welfare) of our mechanism as follows.
\begin{displaymath}
\begin{aligned}
&E_{p_1 \in \mathbb{P}_s}E_{p_2 \in \mathbb{P}_g(p_1)}[W_2(\mathbf{B}, p_1, p_2)]\\
>&E_{p_1 \in \mathbb{P}_s}[OPT_2(p_s) - 3\frac{2ln(e+\frac{\epsilon_2|\mathbb{P}_g(p_s)|OPT_2}{2C})}{\epsilon_2}]\\
>&E_{p_1 \in \mathbb{P}_s}[OPT_2(p_s)] - 3\frac{2ln(e+\frac{\epsilon_2|\mathbb{P}_g(1)|OPT_1}{2C})}{\epsilon_2}\\
>&OPT_1-3(\frac{2C ln(e+\frac{\epsilon_1|\mathbb{P}_s|OPT_1}{2C})}{\epsilon_1}\\
&+\frac{2C ln(e+\frac{\epsilon_2|\mathbb{P}_g(1)|OPT_1}{2C})}{\epsilon_2})\\
=&OPT_1-6(\frac{C ln(e+\frac{\epsilon_1|\mathbb{P}_s|OPT_1}{2C})}{\epsilon_1}\\
&+\frac{C ln(e+\frac{\epsilon_2|\mathbb{P}_g(1)|OPT_1}{2C})}{\epsilon_2})
\end{aligned}
\end{displaymath}
where $|\mathbb{P}_g(1)| = \max_{p \in \mathbb{P}_s} |\mathbb{P}_g(p)|$, and $OPT_1 \ge OPT_2$.

Therefore, DDSM achieves approximate social welfare maximization. 
\end{proof}
From its design, obviously DDSM also achieves other economic properties such as individual rationality and budget balance. Here, we just present the related theorem without proving it.

\begin{mytheorem}
DDSM is of individual rationality and budget balance.
\end{mytheorem}

\section{Improvement of DDSM}\label{sec:improvement}
In the previous section, we introduced our basic mechanism with differential privacy, $\gamma$-truthfulness and approximate social welfare maximization for double spectrum auctions. In this section, we improve the basic mechanism in three aspects, namely auction algorithm design, utility function design, and grouping algorithm extension, enhancing the utility and applicability of our mechanism.

\subsection{Improved Mechanism}

Our main observation is that, the basic mechanism first randomly selects the selling clearing price $p_s$ and then the buying clearing price $p_g$, by employing twice the exponential mechanism, and both selections add roughly the same noise to the mechanism. Then, could we just select a clearing price pair $(p_s, p_g)$ from a properly defined set by employing the exponential mechanism, adding roughly the same noise to the mechanism only once? If so, we can greatly improve the utility of the auction outcome without degrading the privacy level. Following this line of thinking, we design our improved mechanism as follows.

\subsubsection{Design Detail}
\noindent

\textbf{(1) Buyer Group Formation:}
Exactly the same as that of our basic mechanism.

\textbf{(2) Double-sided Price Selection:}
The essential difference of this step from our basic mechanism is that, we view any two clearing prices $p_s$ and $p_g$ as a combination, and randomly select them out together at a time. Here, we mainly state what is different, and leave other things untouched. That is, all other operations that are not detailed, for instance, the group bid computation, and quotation and bid sorting, are just the same as that of our basic mechanism.

In order to select a pair of prices $p_s$ and $p_g$ at a time with the exponential mechanism, we first need to define the output set properly. To ensure that each transaction is profitable, it is necessary that $p_s \le p_g$. Thus, we can define the output set $\mathbb{P}=\{(p_s, p_g)|p_s\in \mathbb{P}_s \wedge p_g\in \mathbb{P}_g(p_s)\}$, where $\mathbb{P}_s=[1..q_{max}]$ and $\mathbb{P}_g(p_s)=[p_s..n_{max}\cdot b_{max}]$.

Next, we define the utility function of the exponential mechanism as $W(\mathbf{q}, \mathbf{b}, p_s, p_g)$, the social welfare resulted from clearing prices $p_s$ and $p_g$. Note that here $W$ can be computed by Eq.~\eqref{eq:social-welfare} just as in our basic mechanism. The sensitivity is $\Delta W = {n_{max}}\cdot{b_{max}} - 1$.

Then, we define a mechanism $\mathcal{M}$ chooses each price pair $(p_s, p_g)$ with a probability exponentially proportional to its corresponding the social welfare:
\begin{displaymath}
Pr[\mathcal{M}(\mathbf{q}, \mathbf{b})=(p_s, p_g)]=\frac{{exp}(\frac{\epsilon W(\mathbf{q}, \mathbf{b}, (p_s, p_g))}{2\Delta W})}{\sum_{(p_s, p_g)'\in\mathbb{P}}{exp}(\frac{\epsilon W(\mathbf{q}, \mathbf{b}, (p_s, p_g)')}{2\Delta W})}
\end{displaymath}
where $\epsilon$ is the privacy budget.

The improved price selection is depicted in Algorithm ~\ref{protocol:price-selection2}.

\renewcommand{\algorithmicrequire}{\textbf{Input:}}  
\renewcommand{\algorithmicensure}{\textbf{Output:}}  

\begin{algorithm}[htb]
\caption{Improved Double-sided Price Selection} \label{protocol:price-selection2}
\begin{algorithmic}[1]

\REQUIRE The set of sellers $\mathbb{M}$, the sellers' quotation profile $\mathbf{q}$, the set of buyers $\mathbb{N}$, the buyers' bid profile $\mathbf{b}$, the buyer grouping $G$, and privacy budgets $\epsilon$.

\ENSURE Selling and buying clearing prices $p_s$ and $p_g$

\vspace{0.5\baselineskip}

Steps (1), (2) and (3) are exactly he same as Algorithm \ref{protocol:price-selection}

(4) \underline{Compute probability distributions and select prices:}

\STATE $\Delta W \leftarrow {n_{max}}\cdot{b_{max}} - 1$

\FOR{$p_s \leftarrow 1$ to $q_{max}$}

\FOR{$p_g \leftarrow p_s$ to $n_{max} \cdot b_{max}$}

\vspace{0.5\baselineskip}

\STATE $Pr[\mathcal{M}(\mathbf{q}, \mathbf{b})=(p_s, p_g)] \leftarrow $ \\$\frac{{exp}(\frac{\epsilon W(\mathbf{q}, \mathbf{b}, (p_s, p_g))}{2\Delta W})}{\sum_{(p_s, p_g)'\in\mathbb{P}}{exp}(\frac{\epsilon W(\mathbf{q}, \mathbf{b}, (p_s, p_g)')}{2\Delta W})}$

\vspace{0.5\baselineskip}

\ENDFOR

\ENDFOR

\STATE $(p_s, p_g) \leftarrow \mathcal{M}(\mathbf{q}, \mathbf{b})$ // \emph{Select clearing price pair}

\RETURN $p_s$ and $p_g$

\end{algorithmic}
\end{algorithm}

\textbf{(3) Auction Outcome Release:}
Given a price pair $(p_s, p_g)$, this step is exactly the same as that of our basic mechanism.

\subsubsection{Analysis}
\noindent

We mainly show that the improve mechanism achieves differential privacy, $\gamma$-truthfulness, and approximate social welfare maximization. Note that the properties of individual rational and budget balance hold similarly for the improved mechanism, but we omit the proofs due to the obviousness.

\begin{mytheorem}\label{the:imp-dp}
The improved mechanism achieves $\epsilon$-differential privacy.
\end{mytheorem}

The proof of Theorem \ref{the:imp-dp} is omitted. The improved mechanism is simply an application of the exponential mechanism, and thus it is certainly $\epsilon$-differential privacy.

\begin{mytheorem}\label{the:imp-gamma-truthfulness}
The improved mechanism is $\gamma$-truthful, where $\gamma=max(\epsilon u_{1max}, \epsilon u_{2max})$.
\end{mytheorem}

\begin{proof}
The proof of Theorem \ref{the:imp-gamma-truthfulness} is similar to that of Theorem \ref{the:gamma-truthfulness}. Here, we only point out some differences, but will not go into detail about the entire proof.

Since the improved mechanism chooses both selling and buying clearing prices at one time, the expected utility of a seller $m$ should become
\begin{displaymath}
\begin{aligned}
&E[u_m(\mathbf{q}, \mathbf{b}, \mathbb{P}_s, \mathbb{P}_g)]\\
=&\sum_{p_1\in\mathbb{P}_s}(\sum_{p_2\in\mathbb{P}_g(p_1)}\alpha(p_1, p_2))\cdot u_m(\mathbf{q}, \mathbf{b}, p_1, p_2)\\
\end{aligned}
\end{displaymath}
where $\alpha(p_1, p_2) = Pr[\mathcal{M}(\mathbf{q}, \mathbf{b})= (p_1,p_2)]$.

Similarly, the expected utility of a buyer $n$ should be
\begin{displaymath}
\begin{aligned}
&E[u_n(\mathbf{q}, \mathbf{b}, \mathbb{P'}_g, \mathbb{P'}_s)]\\
=&\sum_{p_2\in\mathbb{P'}_g}(\sum_{p_1\in\mathbb{P'}_s(p_2)}\alpha(p_1, p_2))\cdot u_n(\mathbf{q}, \mathbf{b}, p_1, p_2)\\
\end{aligned}
\end{displaymath}
where $\mathbb{P'}_g = [1..n_{max} \cdot b_{max}]$, and $\mathbb{P'}_s(p_2) = [1..p_2]$.

With these computations of expected utilities, we can derive that the improved mechanism is $\epsilon u_{1max}$-truthful for sellers and $\epsilon u_{2max}$-truthful for buyers. So the improved mechanism is $\gamma$-truthful, where $\gamma=max(\epsilon u_{1max}, \epsilon u_{2max})$.
\end{proof}

\begin{mytheorem}
The improved mechanism achieves approximate social welfare maximization.
\end{mytheorem}

\begin{proof}
Here we use $p$ to represent tuple $(p_s,p_g)$, and let $\mathbf{B} = \mathbf{q} \cup \mathbf{b}$, $C=n_{max}\cdot b_{max} - 1$. For the improved mechanism, let $OPT=max_{p \in \mathbb{P}}W(\mathbf{B}, p)$ denote the maximum $W$ of mechanism $\mathcal{M}$. We define the sets $R_{t}=\{p:W(\mathbf{B}, p)>OPT-t\}$ and $\overline{R_{2t}}=\{p:W(\mathbf{B}, p)\leq OPT-2t\}$ for a small constant $t>0$. Then, we have
\begin{displaymath}
\begin{aligned}
Pr(\overline{R_{2t}})&\leq \frac{Pr(\overline{R_{2t}})}{Pr(R_{t})}\\
&\leq |\mathbb{P}|{exp}(\frac{-\epsilon t}{2C})\\
\end{aligned}
\end{displaymath}

Then, $Pr(p\in R_{2t})\geq 1-|\mathbb{P}|{exp}(\frac{-\epsilon t}{2C})$. If we have $t\geq \frac{2Cln(\frac{|\mathbb{P}|OPT}{t})}{\epsilon}$, and we make $Pr(p\in R_{2t})\geq 1-\frac{t}{OPT}$.
Therefore, for any $t\geq \frac{2Cln(\frac{|\mathbb{P}|OPT}{t})}{\epsilon}$, we have
\begin{displaymath}
\begin{aligned}
E_{p \in \mathbb{P}}[W(\mathbf{B}, p)]&>\sum_{p\in R_{2t}}W(\mathbf{B}, p)Pr(\mathcal{M}(\mathbf{B})=p)\\
&>(OPT-2t)(1-\frac{t}{OPT})\\
&>OPT-3t
\end{aligned}
\end{displaymath}

If we let $t=\frac{2Cln(e+\frac{\epsilon|\mathbb{P}|OPT}{2C})}{\epsilon}\geq \frac{2C}{\epsilon}$, and we have
\begin{displaymath}
\begin{aligned}
\frac{2ln(\frac{|\mathbb{P}|OPT}{t})}{\epsilon}&\leq \frac{2Cln(e+\frac{\epsilon|\mathbb{P}|OPT}{2C})}{\epsilon}\\
&=t
\end{aligned}
\end{displaymath}

Then put $t$ int the top equation, we have
\begin{displaymath}
\begin{aligned}
E_{p \in \mathbb{P}}[W(\mathbf{B}, p)]&>OPT-3t\\
&>OPT-3\frac{2Cln(e+\frac{\epsilon|\mathbb{P}|OPT}{2C})}{\epsilon}
\end{aligned}
\end{displaymath}

The proof is completed.
\end{proof}

\subsection{Improved Utility Function}\label{sec:improved-utility}

It is quite natural that we use the social welfare as the utility function for the exponential mechanism in DDSM. Our goal is to maximize the social welfare, and selecting prices in term of the social welfare they result in is straightforward to this goal. Nevertheless, we found that the sensitivity of the social welfare (i.e., $n_{max} \cdot b_{max} - 1$) is rarely realized, since it is quite unpractical that buyers in a maximally sized buyer group all bid the maximum bid $b_{max}$, while the corresponding seller asks the minimum quotation $1$. In other word, the sensitivity of the social welfare tends to magnify. Can we find an equivalent utility function to maximize the social welfare but with a ``smaller'' sensitivity? The answer is yes.

Actually, we can simply use the number of seller-buyer-group winner pairs $K$ as the utility function. Intuitively, the social welfare is proportional to the number of winner pairs $K$, and the sensitivity of $K$ is $\Delta K = 1$. Moreover, this sensitivity could be frequently realized, since changing a seller's quotation or a buyer's bid in an auction may easily change the value of $K$ by 1. Thus, from the intuitive discussion above, we can see that $K$ is possibly a good utility function that we can find.

We can confirm the above finding by Theorem \ref{the:utility-k}.
\begin{mytheorem}\label{the:utility-k}
In our mechanism DDSM, it is better to use the number of winner pairs $K$ as the utility function than the social welfare $W$ so as to maximize the social welfare.
\end{mytheorem}

\begin{proof}
Let $\mathcal{M}$ denote the mechanism, $\mathbf{B}$ denote the input of $\mathcal{M}$ and $p$ denote any output of $\mathcal{M}$.

When using the social welfare $W$ as the utility function, we get
\[Pr[\mathcal{M}(\mathbf{B}) = p] \propto exp(\frac{\epsilon W}{2\Delta W})\]

While using the number of winner pairs $K$ as the utility function, we get
\[Pr[\mathcal{M}(\mathbf{B})=p] \propto exp(\frac{\epsilon K}{2\Delta K})\]

Furthermore, let $C = n_{max}\cdot b_{max} - 1$, we have
\begin{displaymath}
\begin{aligned}
exp(\frac{\epsilon W}{2\Delta W}) &= exp(\frac{\epsilon \sum_{l=1}^{K}{w_l}}{2C})\\
&\leq exp(\frac{\epsilon K \cdot C}{2C})\\
&= exp(\frac{\epsilon K}{2\Delta K})\\
&= exp(\frac{\epsilon W}{2\Delta W'})
\end{aligned}
\end{displaymath}
where $\Delta W' \le \Delta W$.

From the inequality above, we can view the utility function $K$ as the utility function $W$ with a smaller sensitivity $\Delta W'$. That is, by applying $K$ as the utility function, we equivalently reduce the sensitivity when applying $W$ as the utility function. Therefore, it is better to use $K$ instead of $W$ as the utility function, since less noise added is beneficial to maximize the social welfare.
\end{proof}
From Theorem \ref{the:utility-k}, we can see that for both the basic mechanism and the improved mechanism, substitution of $W$ with $K$ will probably improve the utility of the auction outcome. Moreover, due to the principle of exponential mechanism, this substitution will not affect the achievement of differential privacy, $\gamma$-truthfulness and approximate social welfare maximization, and so on.

\subsection{Extension to Random Grouping}

In our mechanism, so far, we have required that the buyer grouping algorithm employed should be deterministic and independent of buyers' bids. In other word, the grouping should be completely determined by the buyer set. Here, we show that the grouping algorithm can be extended to any random grouping algorithm independent of buyers' bids, as follows.
\begin{mytheorem}\label{the:random-grouping}
For any random grouping algorithm independent of buyers' bids, DDSM can still achieve $\epsilon$-differential privacy.
\end{mytheorem}

\begin{proof}
Given any random grouping algorithm $\mathcal{G}$, suppose that it can be decomposed into $T$ deterministic grouping algorithms $\mathcal{G}_1, \mathcal{G}_2, \cdots, \mathcal{G}_T$, with probabilities $p_1, p_2, \cdots, p_T$, respectively, where $\sum_{i=1}^T {p_i=1}$. Then, when using each $\mathcal{G}_i (1 \le i \le T)$ as the grouping algorithm in DDSM, we get a mechanism with $\epsilon$-differential privacy due to Theorem \ref{the:dp}. By the convexity of differential privacy as shown in Lemma \ref{the:convexity}, we also get a mechanism with $\epsilon$-differential privacy when using the random grouping $\mathcal{G}$. This completes the proof.
\end{proof}
We can also show that DDSM achieves $\gamma$-truthfulness and approximate social welfare maximization when using a random grouping algorithm, since both properties can be derived from the property of differential privacy as in our previous proofs. Here, we omit these proofs due to the space limitation.

\section{Performance Evaluation}\label{sec:experiment}
In this section, we focus on the implementation of DDSM and the evaluation of its performance.
\subsection{Experimental Settings}
In the experimental settings, we simulate the square area that the buyers are randomly located in is $2000m \times 2000m$, and the distance that any two buyers conflict with each other is $500m$. The bids of buyers and the quotations of sellers are randomly generated in the range [1,50] and [1,100], respectively. The experimental results are averaged over 100 runs. In our experiments, we focus on the following performance metrics:
\begin{itemize}
\item \emph{Social welfare}: The sum of all winning buyers' bids subtracting the sum of all winning sellers' quotations.
\item \emph{Social welfare ratio}: The ratio of the social welfare produced by an auction to the optimal social welfare produced by all possible auctions.
\item \emph{Running time}: The time spent to execute an auction.

\end{itemize}

\subsection{Experimental Results}

In our experiments, we evaluate and compare the performances of the basic DDSM and the improved DDSM. We also use the truthful double spectrum auction mechanism, TRUST\cite{Zhou2009TRUST}, which has no privacy guarantee, as a benchmark mechanism.

\begin{figure}[htbp]
\begin{center}
\includegraphics[width=0.25\textwidth]{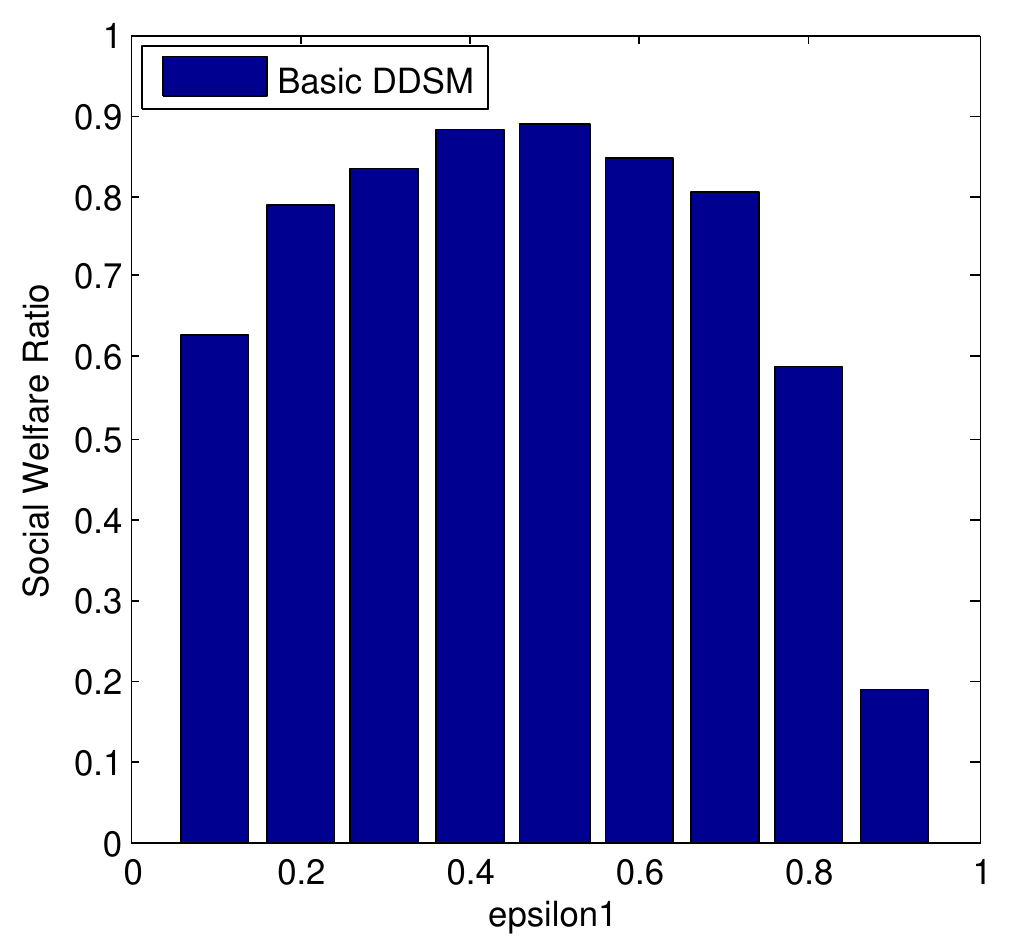}
\caption{Ratio as $\epsilon_1$ of basic DDSM grows ($\epsilon_1+\epsilon_2=1$, $M=800$, $N=200$)}\label{fig:epsilon12}
\end{center}
\end{figure}

\begin{figure*}[ht]
\centering
\begin{minipage}[b]{0.48\linewidth}
\includegraphics[width=1\textwidth]{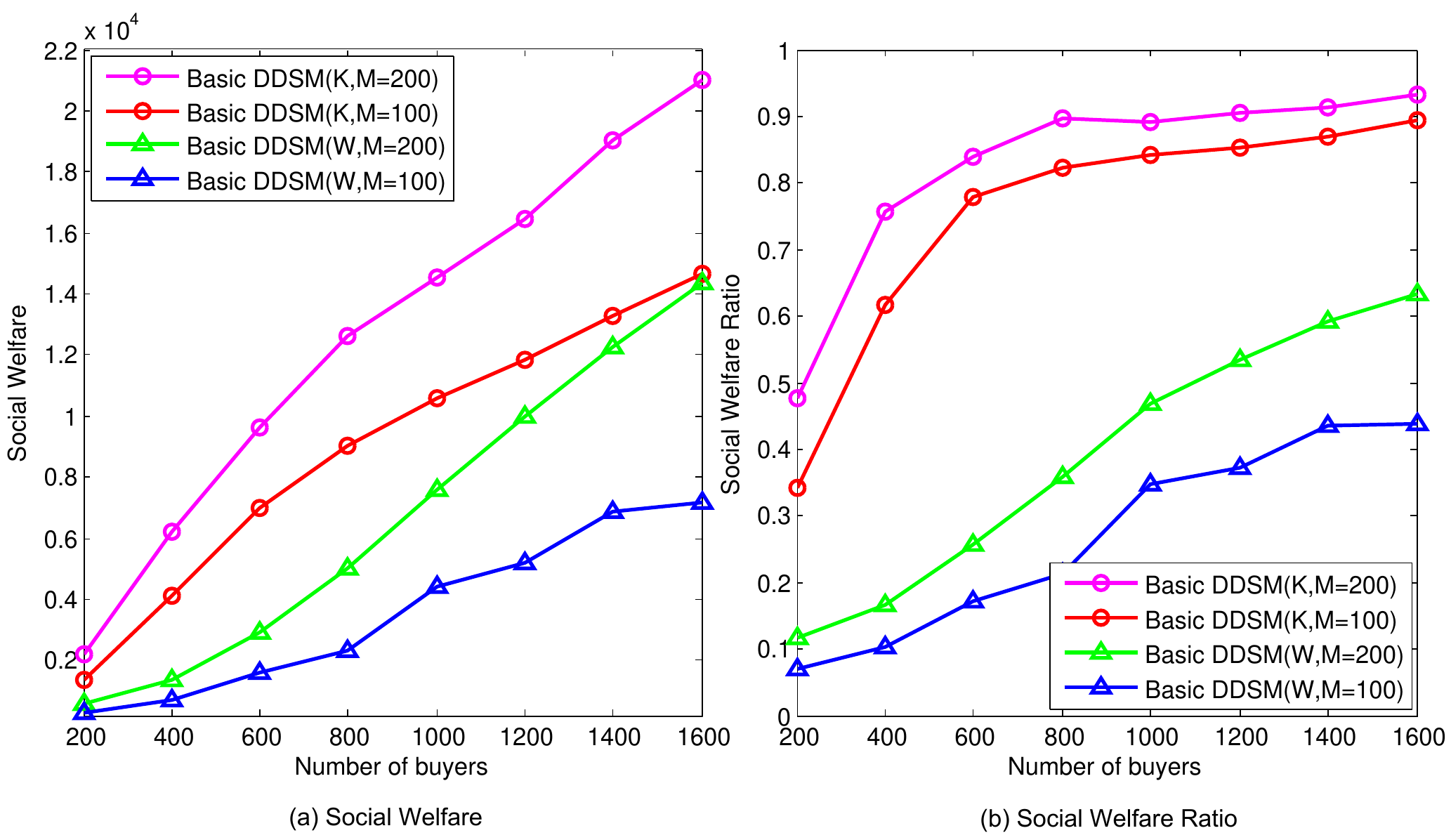}
\caption{Comparison between different utility functions}\label{fig:utility}
\label{fig:minipage1}
\end{minipage}
\quad
\begin{minipage}[b]{0.48\linewidth}
\includegraphics[width=1\textwidth]{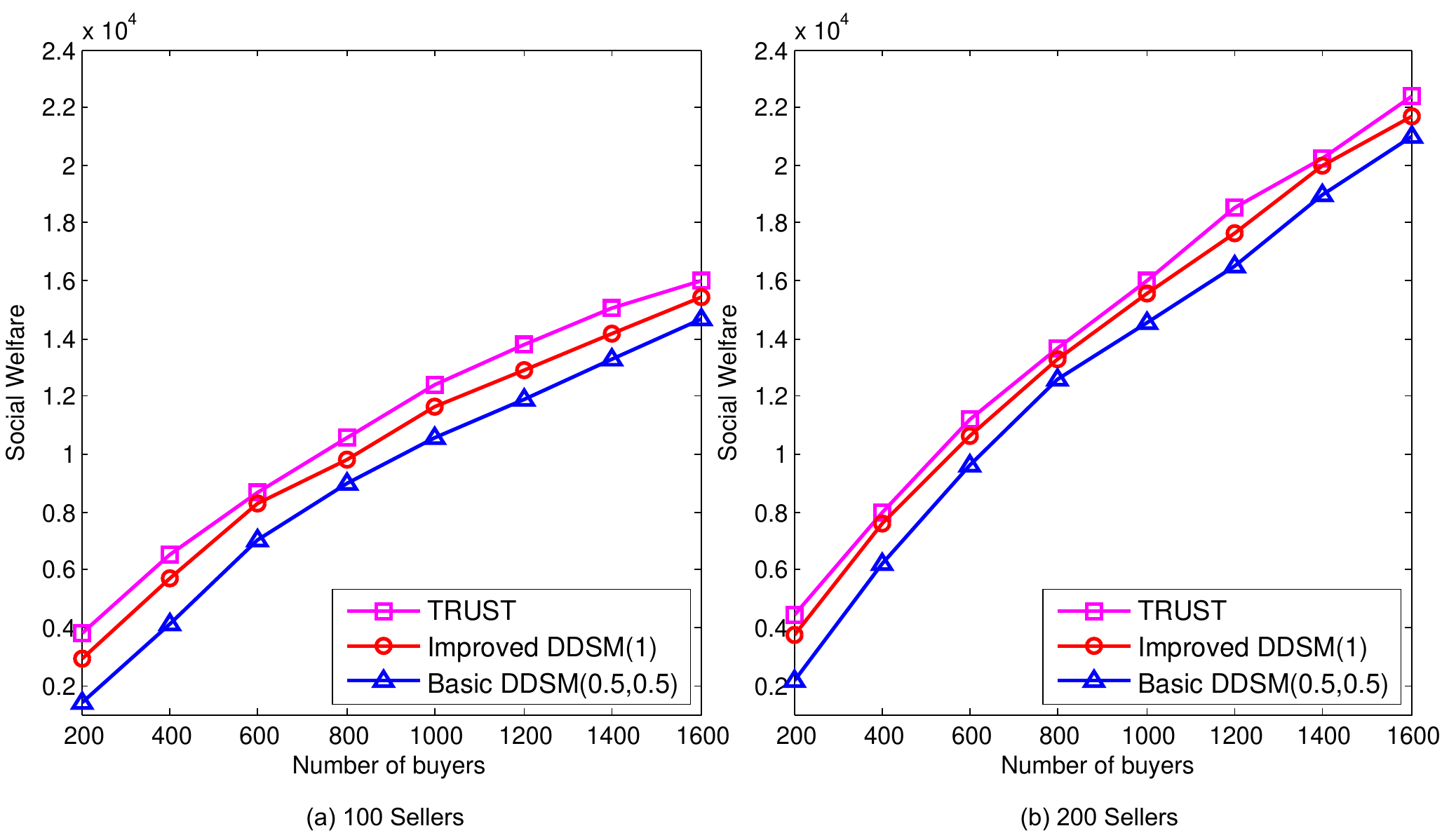}
\caption{Social welfare generated by different mechanisms }\label{fig:cmp-social}
\label{fig:minipage2}
\end{minipage}
\quad
\begin{minipage}[b]{0.48\linewidth}
\includegraphics[width=1\textwidth]{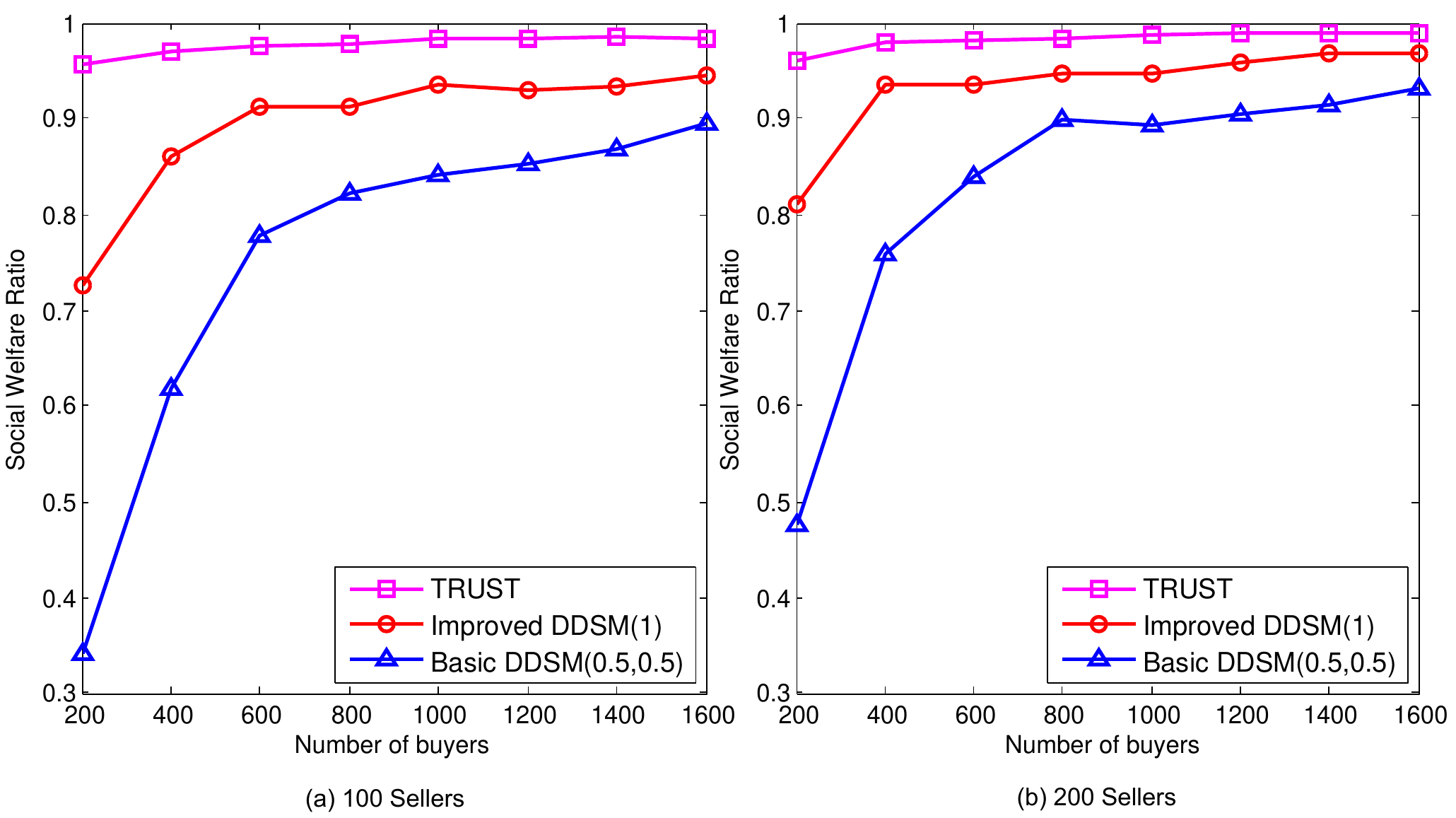}
\caption{Social welfare ratio of different mechanisms}\label{fig:cmp-ratio}
\label{fig:minipage3}
\end{minipage}
\quad
\begin{minipage}[b]{0.48\linewidth}
\includegraphics[width=1\textwidth]{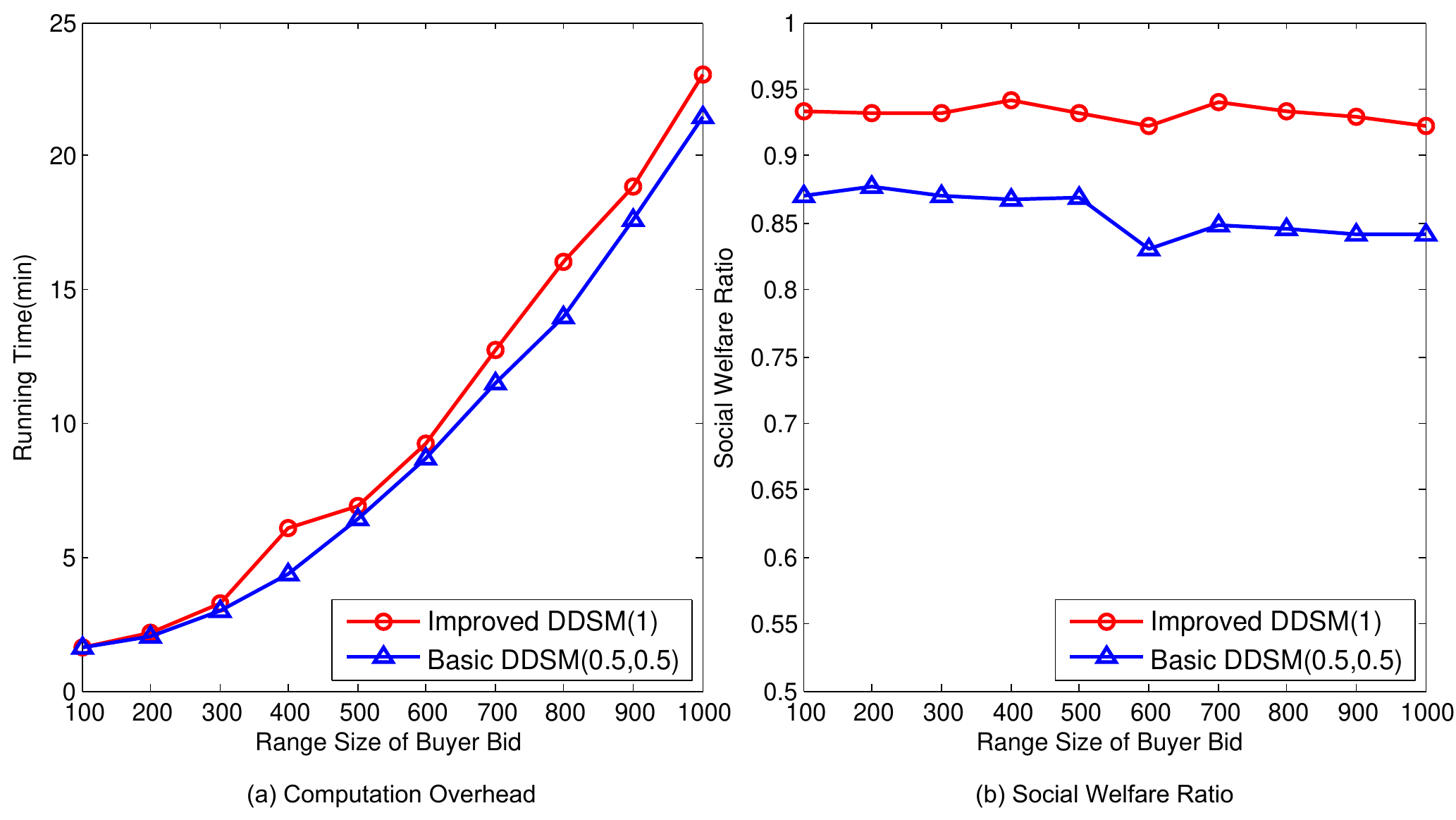}
\caption{Performance comparison of different bid ranges}\label{fig:cmp-range}
\label{fig:minipage4}
\end{minipage}
\end{figure*}

\textit{(1) Performance of basic DDSM.} In Fig.~\ref{fig:epsilon12}, we plot the social welfare ratios of the basic DDSM as the privacy budget $\epsilon_1$ varies from $0.1$ to $0.9$ under the constraint $\epsilon_1 + \epsilon_2 = 1$, while fixing the numbers of sellers and buyers to $800$ and $200$, respectively. We can see that the social welfare ratio reaches the maximum when $\epsilon_1 = \epsilon_2 = 0.5$, and degrades quickly when $\epsilon_1\gg \epsilon_2$. This shows that both privacy budgets $\epsilon_1$ and $\epsilon_2$ are roughly of equal significance for achieving utility, so they should be assigned the same value. Furthermore, $\epsilon_2$ is more sensitive to small values, which would add much noise to the second exponential mechanism, and make it produce a poor final social welfare. Overall, the maximum social welfare ratio is nearly $0.9$, and it seems acceptable for practice. In the following, we will use the basic DDSM with $\epsilon_1 = \epsilon_2 = 0.5$ unless otherwise stated.

\textit{(2) Performance of improved utility function.} Fig.~\ref{fig:utility} illustrates the comparisons of both the social welfare ratios and the social welfare values of the basic DDSM when using different utility functions, namely the social welfare $W$ and the number of winner pairs $K$. From Fig.~\ref{fig:utility}(a), we can see that the resulted social welfare when using utility function $K$ is obviously higher than that when using utility function $W$, and this gap becomes even larger as the number of buyers or sellers increases. Similar trends can be observed even more apparently for social welfare ratios in Fig.~\ref{fig:utility}(b). This experimental result demonstrates that utility function $K$ provide a better utility and thus a higher social welfare value than utility function $W$, which is in accord with our theoretical result in Section~\ref{sec:improved-utility}. In the following experiments, we will use the number of winner pairs $K$ as the utility function for both the basic DDSM and the improved DDSM.

\textit{(3) Performance comparison: Basic vs. Improved.} Fig.~\ref{fig:cmp-social} traces the social welfare values of the basic DDSM, the improved DDSM, and TRUST, when the number of buyers $N$ varies from 200 to 1600 with at the step of 200, and the number of sellers $M$ is fixed to 100 and 200. Expectedly, the social welfare values of all three mechanisms increase as the number of buyers or sellers increases, and TRUST performs the best, the improved DDSM takes second place, and the basic DDSM is the last one. Moreover, the curve of the improved DDSM is much closer to that of TRUST than that of the basic DDSM. Since TRUST guarantees truthfulness but no differential privacy, we view TRUST as an ideal benchmark mechanism to struggle for. The above observations show that the improved DDSM well outperforms the basic DDSM, and performs even closely to ideal mechanism of TRUST.

Fig.~\ref{fig:cmp-ratio} traces the social welfare ratios of the basic DDSM, the improved DDSM and TRUST in exactly the same conditions as that of Fig.~\ref{fig:cmp-social}. As the number of buyers $N$ raises, the social welfare ratio of TRUST remains basically consistent and always more than $90\%$, while the ratios of the basic DDSM and the improved DDSM are getting closer and closer to that of TRUST. Similarly, the curve of the improved DDSM is much closer to that of TRUST than that of the basic DDSM, and this means that the improved DDSM also well outperforms the basic DDSM in social welfare ratios.

Fig.~\ref{fig:cmp-range} shows the comparisons of running times and social welfare ratios between the basic DDSM and the improved DDSM as the range size of buyers' bids (the range size of sellers' quotations are twice of it) varies from $100$ to $1000$, while the numbers of buyers and sellers are fixed to 800 and 200, respectively.  From Fig.~\ref{fig:cmp-range}(a), we can observe that the improved DDSM takes slightly more running time than the basic DDSM. The reason is that the size of the sampling set in the improved DDSM is much bigger than that in the basic DDSM. Specifically, the former is the product of the range size of buyers' bids and that of sellers' quotations, while the later is the sum of these two range sizes. From Fig.~\ref{fig:cmp-range}(b), we find that the social welfare ratio of the improved DDSM is always higher than that of the basic DDSM as the range size increases, and this indicates that the increase of the range size has limited impact on the performance advantage of the improved DDSM over the basic DDSM.

\begin{figure}[htbp]
\begin{center}
\includegraphics[width=0.25\textwidth]{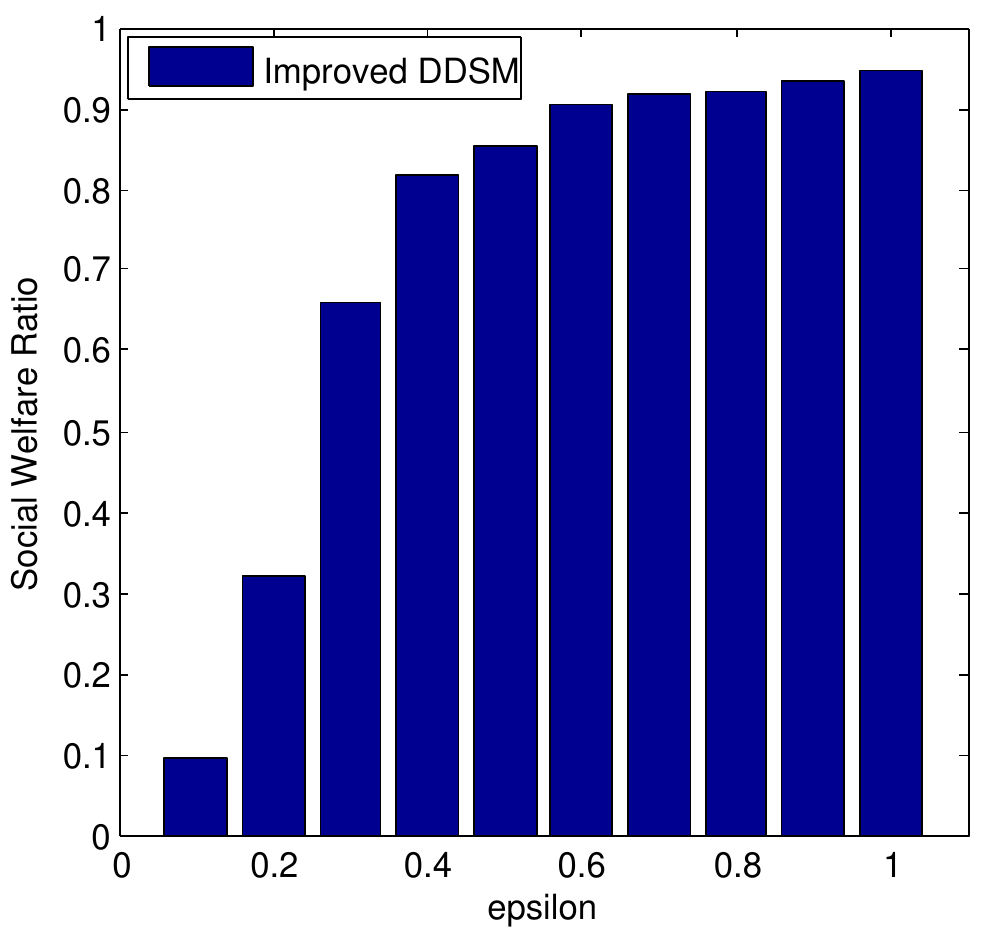}
\caption{Ratio as $\epsilon$ of the improved DDSM grows}\label{fig:epsilon}
\end{center}
\end{figure}

\textit{(4) Performance of improved DDSM.} Fig.~\ref{fig:epsilon} depicts the performance of the improved DDSM in term of social welfare ratio as the privacy budget $\epsilon$ varies from $0.1$ to $1$, while the number of buyers $N=800$ and the number of sellers $N=200$. Naturally, the social welfare ratio raises as the privacy budget $\epsilon$ increases and thus the privacy lever lowers. We can see that the social welfare ratio becomes more than $0.9$ when $\epsilon > 0.5$. The performance could be acceptable for practical applications.

\section{Conclusions}\label{sec:conclusions}
In this paper, we have proposed a differentially private double spectrum auction mechanism, DDSM, which also achieves $\gamma$-truthfulness and approximate social welfare maximization. To achieve our design goal, we employ the exponential mechanism twice to randomly select selling and buying clearing prices, respectively, and we theoretically prove that all the three properties above-mentioned hold. Later, we improve our mechanism by employing the exponential mechanism only once to reduce the noise added to the mechanism, by designing a utility function with a smaller sensitivity, and by extending deterministic buyer grouping algorithm to random ones. Both theoretical analysis and experimental evaluations show that the improved mechanism achieves a better utility compared to the basic mechanism for the same privacy level.

\bibliography{sample-bibliography}{}
\bibliographystyle{IEEEtran}



%

%
%
%

\end{document}